\documentclass{birkjour}
\usepackage[colorlinks=true]{hyperref}
 \newtheorem{thm}{Theorem}[section]
 \newtheorem*{thm*}{Theorem}
 \newtheorem{cor}[thm]{Corollary}
 \newtheorem{lem}[thm]{Lemma}
 \newtheorem{prop}[thm]{Proposition}
 \theoremstyle{definition}
 \newtheorem{defn}[thm]{Definition}
 \theoremstyle{remark}
 \newtheorem{rem}[thm]{Remark}
 
 \numberwithin{equation}{section}

\newcommand{\bra}[1]{\left(#1\right)}
\newcommand{\sqa}[1]{\left[#1\right]}
\newcommand{\norm}[1]{\left\lVert#1\right\rVert}

\begin{document}
\title[One-mode quatum Gaussian channels have Gaussian maximizers]{The one-mode quantum-limited Gaussian attenuator and amplifier have Gaussian maximizers}
\author{Giacomo De Palma}
\address{QMATH, Department of Mathematical Sciences, University of Copenhagen, Universitetsparken 5, 2100 Copenhagen, Denmark\\
NEST, Scuola Normale Superiore and Istituto Nanoscienze-CNR, I-56126 Pisa, Italy\\
INFN, Pisa, Italy}
\email{giacomo.depalma@math.ku.dk}
\author{Dario Trevisan}
\address{Universit\`a degli Studi di Pisa, I-56126 Pisa, Italy}
\email{dario.trevisan@unipi.it}
\author{Vittorio Giovannetti}
\address{NEST, Scuola Normale Superiore and Istituto Nanoscienze-CNR, I-56126 Pisa, Italy}
\email{vittorio.giovannetti@sns.it}



\date{}

\begin{abstract}
We determine the $p\to q$ norms of the Gaussian one-mode quantum-limited attenuator and amplifier and prove that they are achieved by Gaussian states, extending to noncommutative probability the seminal theorem ``Gaussian kernels have only Gaussian maximizers'' (Lieb, Invent. Math. \textbf{102}, 179 (1990)).
The quantum-limited attenuator and amplifier are the building blocks of quantum Gaussian channels, which play a key role in quantum communication theory since they model in the quantum regime the attenuation and the noise affecting any electromagnetic signal.
Our result is crucial to prove the longstanding conjecture stating that Gaussian input states minimize the output entropy of one-mode phase-covariant quantum Gaussian channels for fixed input entropy.
Our proof technique is based on a new noncommutative logarithmic Sobolev inequality, and it can be used to determine the $p\to q$ norms of any quantum semigroup.
\end{abstract}
\maketitle

\section{Introduction}
Given $p,\,q > 1$, let us consider a real Gaussian integral kernel $G$ from $L^p(\mathbb{R}^m)$ to $L^q(\mathbb{R}^n)$:
\begin{equation}
(G\,f)(x) = \int_{\mathbb{R}^m} G(x,y)\,f(y)\,d^my\;,\qquad x\in\mathbb{R}^n\;,\qquad f\in L^p(\mathbb{R}^m)\;,
\end{equation}
where $G(x,y)$ is a real Gaussian function on $\mathbb{R}^{m+n}$, i.e., the exponential of a quadratic polynomial in $x$ and $y$ with real coefficients.
The $p\to q$ norm of $G$ is
\begin{equation}\label{eq:defpq}
\left\|G\right\|_{p\to q}:=\sup_{0<\|f\|_p<\infty}\frac{\left\|G\,f\right\|_q}{\left\|f\right\|_p}\;.
\end{equation}
In the seminal paper ``Gaussian kernels have only Gaussian maximizers'' \cite{lieb1990Gaussian}, E. H. Lieb proved that the determination of the supremum in \eqref{eq:defpq} can be restricted to real Gaussian functions.
If this supremum is finite, it is attained on a real Gaussian function $f$; if it is infinite, it is asymptotically attained by a suitable sequence of real Gaussian functions.
This result permits to determine the $p\to q$ norms of $G$ and has countless applications, such as straightforward proofs of the Brascamp--Lieb convolution inequality, the Hausdorff--Young--Titchmarsh inequality for Fourier integrals and Nelson's hypercontractivity theorem \cite{lieb1990Gaussian}, a proof of the Entropy Power Inequality and of the Brunn--Minkowski Inequality (see e.g. \cite[section 17.8]{cover2006elements}), and  Lieb's solution \cite{lieb1978proof,lieb2014proof} of Wehrl's conjecture \cite{wehrl1979relation,anderson1993information}, stating that coherent states minimize the Shannon differential entropy of the Husimi $Q$ representation.

In noncommutative probability, functions on $\mathbb{R}^n$ with $n$ even are replaced by operators acting on the Hilbert space $\mathcal{H}$ of an $n/2$-mode Gaussian quantum system, i.e. the irreducible representation of the canonical commutation relations of the ladder operators (see e.g. \cite[Chapter 12]{holevo2013quantum} or \cite{serafini2017quantum})
\begin{equation}\label{eq:CCR}
\left[\hat{a}_i,\;\hat{a}^\dag_j\right]=\delta_{ij}\,\hat{\mathbb{I}}\;,\qquad\left[\hat{a}_i,\;\hat{a}_j\right]=0\;,\qquad i,\,j=1,\ldots,\,\frac{n}{2}\;.
\end{equation}
The $L^p$ norm of a function $f:\mathbb{R}^n\to\mathbb{C}$ is replaced by the Schatten $p$ norm \cite{schatten1960norm,holevo2006multiplicativity} of a linear operator $\hat{X}:\mathcal{H}\to\mathcal{H}$, defined as the $l^p$ norm of its singular values:
\begin{equation}\label{eq:defpn}
\left\|\hat{X}\right\|_p:=\left(\mathrm{Tr}\left(\hat{X}^\dag\hat{X}\right)^\frac{p}{2}\right)^\frac{1}{p}\;.
\end{equation}
Integral kernels are replaced by linear maps acting on the operators on $\mathcal{H}$.
For any $p,\,q\ge1$, the $p\to q$ norm of any such map $\Phi$ is \cite{holevo2006multiplicativity}
\begin{equation}\label{eq:defCpq}
\left\|\Phi\right\|_{p\to q}:=\sup_{0<\left\|\hat{X}\right\|_p<\infty}\frac{\left\|\Phi\left(\hat{X}\right)\right\|_q}{\left\|\hat{X}\right\|_p}\;,
\end{equation}
and it can be either finite or infinite.

Quantum Gaussian channels \cite{holevo2013quantum,holevo2015gaussian} are the noncommutative counterpart of Gaussian integral kernels.
They play a key role in quantum communication theory since they model in the quantum regime the attenuation and the noise that unavoidably affect any electromagnetic communication through metal wires, optical fibers or free space (see e.g. \cite{caves1994quantum,braunstein2005quantum,weedbrook2012gaussian} and references therein).
Quantum Gaussian channels have been conjectured to have Gaussian maximizers since 2006 \cite{holevo2006multiplicativity}: ``In classical information theory the Gaussian channels admit Gaussian maximizers; moreover, there are corresponding analytic results for norms of integral operators with Gaussian kernels. The problem of whether or not there is an analogue of this property for bosonic Gaussian channels is another open question which deserves a separate discussion.''

We prove this longstanding conjecture for the two building blocks of one-mode phase-covariant quantum Gaussian channels: the one-mode Gaussian quantum-limited attenuator and amplifier.
We prove for these channels that for any $1< p < q$ the supremum in \eqref{eq:defCpq} is achieved by a quantum Gaussian operator (Theorems \ref{thm:pq} and \ref{thm:ampl}), i.e., an operator proportional to the exponential of a quadratic polynomial in the ladder operators \eqref{eq:CCR}.
Our result implies an upper bound to the $p\to q$ norms of any one-mode phase-covariant quantum Gaussian channel (Proposition \ref{prop:bound}), and we conjecture that this upper bound is actually optimal.

So far the conjecture ``quantum Gaussian channels have Gaussian maximizers'' has been proven only for $p=1$ \cite{holevo2015gaussian} and $p=q$ \cite{frank2017norms}.
The proof for $p=q$ follows from complex interpolation.
The proof for $p=1$ follows from the proof of the Gaussian majorization conjecture \cite{giovannetti2015majorization,mari2014quantum}, stating that for any phase-covariant quantum Gaussian channel the output generated by the vacuum input state majorizes the output generated by any other positive operator with unit trace.

We also prove that for any $1<q<p$ the $p\to q$ norm of both the quantum-limited attenuator and amplifier is infinite and it is asymptotically achieved by a sequence of Gaussian operators converging to the identity.
The same sequence asymptotically achieves the $p\to p$ norm for any $p>1$.
Hence, our results imply that the semigroups associated to the generators of the quantum-limited attenuator and amplifier are not hypercontractive.
However, with a suitable different definition of the norms, hypercontractivity holds for the semigroup associated to the generator of any one-mode phase-covariant quantum Gaussian channel admitting a stationary state \cite{carbone2008hypercontractivity}.

As in the classical case, the determination of the noncommutative $p\to q$ norms requires the maximization of a convex function over a convex set, that is highly nontrivial since the tools of convex analysis cannot be applied.
Our proof starts from a recent majorization result on one-mode quantum Gaussian channels \cite{de2015passive}, that reduces the problem to input operators diagonal in the Fock basis.
Our proof technique is completely new in the field of noncommutative probability.
Our main result is that Gaussian states achieve the $p\to q$ norms of the quantum-limited attenuator (Theorem \ref{thm:pq}).
The key point to prove Theorem \ref{thm:pq} is reducing the claim to a new noncommutative logarithmic Sobolev inequality (Theorem \ref{thm:logs}).
This inequality provides an upper bound on the derivative of the norms of the output of the attenuator with respect to the attenuation coefficient.
The reduction of Theorem \ref{thm:pq} to Theorem \ref{thm:logs} is inspired by a seminal paper by L. Gross \cite{gross1993logarithmic}, which exploits a logarithmic Sobolev inequality to determine the $p\to q$ norms of classical Gaussian integral kernels.
This is the first time that this technique is exploited in the noncommutative setting.
The difficulty is increased since the optimization of the $p\to q$ norms over quantum Gaussian operators cannot be performed analytically, and a closed formula for these norms cannot be provided.
The proof of Theorem \ref{thm:logs} exploits the Karush--Kuhn--Tucker conditions for constrained local optimizers \cite{borwein2013convex,kuhn1951}, and proceeds along the same lines as the proof of the isoperimetric inequality in \cite{de2016gaussian}.
We then determine the $p\to q$ norms of the quantum-limited amplifier (Theorem \ref{thm:ampl}) exploiting that its Hilbert-Schmidt dual is proportional to the quantum-limited attenuator.
This is the first time that this duality is exploited to prove entropic inequalities.

A fundamental application of our determination of the $p\to q$ norms of the one-mode quantum-limited amplifier is the proof \cite{de2016gaussiannew} of the constrained minimum output entropy conjecture \cite{guha2008entropy,qi2016capacities,giovannetti2010generalized}, which was open since 2008 and states that Gaussian input states minimize the output von Neumann entropy of one-mode phase-covariant quantum Gaussian channels for fixed input entropy.
Quantum states are positive operators with unit trace.
They are the noncommutative counterparts of probability distributions on $\mathbb{R}^n$.
The von Neumann entropy of a quantum state $\hat{\rho}$ is the Shannon entropy of the probability distribution associated to its eigenvalues, i.e.
\begin{equation}
S\left(\hat{\rho}\right)=-\mathrm{Tr}\left[\hat{\rho}\ln\hat{\rho}\right]\;,
\end{equation}
and plays a key role in quantum information and communication theory (see e.g. \cite{wilde2017quantum}) analogous to the role of the Shannon entropy in classical information and communication theory.
A Gaussian quantum state is a Gaussian operator that is also a quantum state.
The constrained minimum output entropy conjecture is crucial to determine both the triple trade-off region and the capacity region for broadcast communication of the Gaussian quantum-limited attenuator and amplifier \cite{guha2007classical,guha2007classicalproc,qi2016capacities}.
So far, the conjecture had been proven only for the Gaussian one-mode quantum-limited attenuator \cite{de2016gaussian} or for zero input entropy for any phase-covariant quantum Gaussian channel \cite{giovannetti2015solution,holevo2015gaussian}.
The constrained minimum output entropy of quantum Gaussian channels has been bounded by the quantum Entropy Power Inequality \cite{konig2014entropy,de2014generalization,de2015multimode,de2017gaussian}.
However, this bound is optimal only when the input state has the same entropy as the state of the environment \cite{de2014generalization}, hence it is not sufficient to prove the conjecture.
The proof of the conjecture for the one-mode quantum-limited Gaussian attenuator of \cite{de2016gaussian} is based on an isoperimetric inequality that is also proven through the Karush--Kuhn--Tucker conditions.
The extension of this isoperimetric inequality to any one-mode phase-covariant quantum Gaussian channel would imply the conjecture for these channels \cite{qi2017minimum}.
However, the Karush--Kuhn--Tucker conditions are sufficient to prove the isoperimetric inequality only for the quantum-limited attenuator, since for any other one-mode phase-covariant quantum Gaussian channel these conditions admit other solutions than quantum Gaussian states \cite{qi2017minimum}.
This problem has motivated the exploration of the $p\to q$ norms to prove the conjecture.
We give a sketch of the proof of the constrained minimum output entropy conjecture via the $p\to q$ norms of the amplifier in \autoref{secMOE}.
For a comprehensive presentation of all the proven or conjectured entropic inequalities for quantum Gaussian channels, we refer the reader to the review \cite{de2018gaussian}.

The restriction of the one-mode quantum-limited attenuator to input operators diagonal in the Fock basis is the linear map acting on discrete classical probability distributions on $\mathbb{N}$ known in the probability literature under the name of thinning \cite{de2015passive}.
The thinning has been introduced by R\'enyi \cite{renyi1956characterization} as a discrete analogue of the rescaling of a continuous real random variable.
The thinning has played this role in discrete versions of the central limit theorem \cite{harremoes2007thinning,yu2009monotonic,harremoes2010thinning}
and of the Entropy Power Inequality \cite{yu2009concavity,johnson2010monotonicity}.
Our result implies that for any $1<p<q$ the $l^p\to l^q$ norm of the thinning is achieved by some geometric probability distribution on $\mathbb{N}$ (Theorem \ref{thm:pqT}).
Moreover, for any $1<q<p$ the $l^p\to l^q$ norm of the thinning is infinite and is asymptotically achieved by a sequence of geometric probability distributions with the ratio converging to $1$.
The same sequence asymptotically achieves the $l^p\to l^p$ norm for any $p>1$.

The paper is structured as follows.
In \autoref{sec:setup} we introduce the one-mode quantum-limited attenuator and amplifier.
In \autoref{sec:logs} we prove the logarithmic Sobolev inequality for the quantum-limited attenuator, and in \autoref{sec:pq} we apply the result to determine the $p\to q$ norms of the quantum-limited attenuator.
In \autoref{sec:ampl} we determine the $p\to q$ norms of the quantum-limited amplifier, and in \autoref{sec:thermal} we determine the upper bound to the $p\to q$ norms of the thermal channels.
In \autoref{secMOE} we sketch the proof of the constrained minimum output entropy conjecture.
The relation of the $p\to q$ norms of the attenuator with the thinning is discussed in \autoref{secthinning}.
The conclusions are in \autoref{sec:concl}.

\section{Setup}\label{sec:setup}
\subsection{Gaussian quantum systems}
We consider the Hilbert space of one harmonic oscillator, or one mode of electromagnetic radiation, i.e. the irreducible representation of the canonical commutation relation (see \cite[Chapter 12]{holevo2013quantum} or \cite{serafini2017quantum} for a more complete presentation)
\begin{equation}
\left[\hat{a},\;\hat{a}^\dag\right]=\hat{\mathbb{I}}\;.
\end{equation}
The operator $\hat{a}$ is called ladder operator.
We define the Hamiltonian $\hat{N}=\hat{a}^\dag\hat{a}$, that counts the number of excitations, or photons.
The vector annihilated by $\hat{a}$ is the vacuum $|0\rangle$, from which the Fock states are built:
\begin{equation}\label{fock}
|n\rangle=\frac{\left(\hat{a}^\dag\right)^n}{\sqrt{n!}}|0\rangle\;,\quad\langle m|n\rangle=\delta_{mn}\;,\quad \hat{N}|n\rangle=n|n\rangle\;,\quad m,\,n\in\mathbb{N}\;.
\end{equation}
An operator diagonal in the Fock basis is called Fock-diagonal.

\subsection{Quantum Gaussian states}
An operator proportional to the exponential of a quadratic polynomial in $\hat{a}$ and $\hat{a}^\dag$ is a Gaussian operator.
If the operator is also positive and has unit trace, it is a quantum Gaussian state.
If the polynomial is proportional to the Hamiltonian $\hat{a}^\dag\hat{a}$, the Gaussian state is thermal, and corresponds to a geometric probability distribution for the energy:
\begin{equation}\label{eq:omegaz}
\hat{\omega}_z=\sum_{n=0}^\infty \left(1-z\right)z^n\;|n\rangle\langle n| = \left(1-\mathrm{e}^{-\beta}\right)\mathrm{e}^{-\beta \hat{a}^\dag\hat{a}}\;,
\end{equation}
where $\beta>0$ is the inverse temperature and $z = \mathrm{e}^{-\beta}$.
The average energy of $\hat{\omega}_z$ is
\begin{equation}
E(z):=\mathrm{Tr}\left[\hat{N}\,\hat{\omega}_z\right]=\frac{z}{1-z}\;,
\end{equation}
and its von Neumann entropy is
\begin{equation}\label{eq:Sz}
S(z)=-\ln\left(1-z\right)-\frac{z\ln z}{1-z}\;.
\end{equation}

\subsection{Gaussian quantum-limited attenuator}
The Gaussian quantum-limited attenuator $\mathcal{E}_\lambda$ of transmissivity $0\le\lambda\le 1$ mixes the input operator $\hat{X}$ with the vacuum state of an ancillary quantum system $B$ through a beamsplitter of transmissivity $\lambda$.
The beamsplitter is implemented by the unitary operator
\begin{equation}\label{defU}
\hat{U}_\lambda=\exp\left(\left(\hat{a}^\dag\hat{b}-\hat{a}\,\hat{b}^\dag\right)\arccos\sqrt{\lambda}\right)\;,
\end{equation}
that satisfies
\begin{equation}
\hat{U}_\lambda^\dag\;\hat{a}\;\hat{U}_\lambda=\sqrt{\lambda}\;\hat{a}+\sqrt{1-\lambda}\;\hat{b}\;,
\end{equation}
where $\hat{b}$ is the ladder operator of the ancilla system $B$ (see  \cite[section 1.4.2]{ferraro2005gaussian}), and
\begin{equation}
\mathcal{E}_\lambda\left(\hat{X}\right)=\mathrm{Tr}_B\left[\hat{U}_\lambda\left(\hat{X}\otimes |0\rangle_B\langle0|\right)\hat{U}_\lambda^\dag\right]\;.
\end{equation}

The quantum-limited attenuator preserves the set of Fock-diagonal operators \cite{ivan2011operator}.
If the input is a Fock-diagonal quantum state, i.e. it has definite photon number, $\mathcal{E}_\lambda$ lets each photon be transmitted with probability $\lambda$ and reflected or absorbed with probability $1-\lambda$, hence the name ``quantum-limited attenuator''.

From \cite[Lemma 13]{de2015passive}, the quantum-limited attenuators form a semigroup with generator
\begin{equation}\label{eq:defL}
\mathcal{L}\left(\hat{X}\right)=\hat{a}\,\hat{X}\,\hat{a}^\dag-\frac{1}{2}\left\{\hat{a}^\dag\hat{a},\;\hat{X}\right\}\;,
\end{equation}
i.e.\ $\mathcal{E}_{\mathrm{e}^{-t}}=\mathrm{e}^{t\mathcal{L}}$ for any $t\ge0$.

The quantum-limited attenuator sends thermal states into themselves, i.e.\ $\mathcal{E}_\lambda\left(\hat{\omega}_z\right)=\hat{\omega}_{z'}$, with
\begin{equation}\label{eq:zlambda}
z'=\frac{\lambda\,z}{1-\left(1-\lambda\right)z}\;.
\end{equation}

\subsection{Gaussian quantum-limited amplifier}
The Gaussian quantum-limited amplifier $\mathcal{A}_\kappa$ of amplification parameter $\kappa\ge1$ performs a two-mode squeezing \cite{barnett2002methods} on the input operator $\hat{X}$ and the vacuum state of an ancillary quantum system $B$.
The squeezing is implemented by the unitary operator
\begin{equation}\label{defUk}
\hat{U}_\kappa=\exp\left(\left(\hat{a}^\dag\hat{b}^\dag-\hat{a}\,\hat{b}\right)\mathrm{arccosh}\sqrt{\kappa}\right)\;,
\end{equation}
that satisfies
\begin{equation}
\hat{U}_\kappa^\dag\;\hat{a}\;\hat{U}_\kappa=\sqrt{\kappa}\;\hat{a}+\sqrt{\kappa-1}\;\hat{b}^\dag\;,
\end{equation}
where $\hat{b}$ is the ladder operator of the ancilla system $B$ (see \cite[section 1.4.4]{ferraro2005gaussian}), and
\begin{equation}
\mathcal{A}_\kappa\left(\hat{X}\right)=\mathrm{Tr}_B\left[\hat{U}_\kappa\left(\hat{X}\otimes |0\rangle_B\langle0|\right)\hat{U}_\kappa^\dag\right]\;.
\end{equation}
Also the quantum-limited amplifier preserves the set of Fock-diagonal operators \cite{ivan2011operator}.

\begin{prop}[{\cite[Theorem 9]{ivan2011operator}}]\label{prop:dual}
For any $\kappa\ge1$ the quantum-limited amplifier $\mathcal{A}_\kappa$ and the quantum-limited attenuator $\mathcal{E}_\frac{1}{\kappa}$ are mutually dual, i.e.\ for any trace-class operator $\hat{X}$ and any bounded operator $\hat{Y}$
\begin{align}\label{eq:dual}
\kappa\;\mathrm{Tr}\left[\hat{Y}\;\mathcal{A}_\kappa\left(\hat{X}\right)\right] &= \mathrm{Tr}\left[\mathcal{E}_\frac{1}{\kappa}\left(\hat{Y}\right)\;\hat{X}\right]\\
\mathrm{Tr}\left[\hat{Y}\;\mathcal{E}_\kappa\left(\hat{X}\right)\right] &= \kappa\;\mathrm{Tr}\left[\mathcal{A}_\kappa\left(\hat{Y}\right)\;\hat{X}\right]\;.
\end{align}
\end{prop}

The quantum-limited amplifier sends thermal states into themselves:
\begin{equation}\label{eq:z'}
\mathcal{A}_\kappa\left(\hat{\omega}_z\right)=\hat{\omega}_{z'}\;,\qquad z'=\frac{z+\kappa-1}{\kappa}\;.
\end{equation}

\section{The logarithmic Sobolev inequality for the quantum-limited attenuator}\label{sec:logs}
\begin{thm}[Logarithmic Sobolev inequality]\label{thm:logs}
We define for any $0<a<1$, any $0<z<1$ and any positive operator $\hat{X}$ with finite rank
\begin{align}\label{eq:defs}
F_a\left(\hat{X}\right) &:= \left.\frac{d}{dt} \ln \norm{\mathrm{e}^{t\mathcal{L}}\left(\hat{X}\right)}_{\frac 1 a } \right|_{t=0}\;,\nonumber\\
S_a\left(\hat{X}\right) &:= S\left(\frac{\hat{X}^{1/a}}{\mathrm{Tr}\,\hat{X}^{1/a}}\right)\;,\nonumber\\
\mathcal{F}_{z,a}\left(\hat{X}\right) &:= F_a\left(\hat{X}\right) + \mu(z,a)\,S_a\left(\hat{X}\right) \;,
\end{align}
where $\mathcal{L}$ is the generator of the quantum-limited attenuator \eqref{eq:defL} and
\begin{equation}\label{eq:mu00}
\mu(z,a):=\frac{a\,z^\frac{a-1}{a}-1+\left(1-a\right)z}{\ln z^\frac{1}{a}}\;.
\end{equation}
Then,
\begin{equation}\label{logs}
\mathcal{F}_{z,a}\left(\hat{X}\right) \le \mathcal{F}_{z,a}\left(\hat{\omega}_z\right)\;.
\end{equation}
Moreover,
\begin{equation}\label{eq:Fa}
F_a\left(\hat{X}\right) \le 1-a\;.
\end{equation}
\end{thm}
An interesting consequence of Theorem \ref{thm:logs} is the following.
\begin{cor}
For any $0<a<1$ and any positive operator $\hat{X}$ with finite rank
\begin{equation}
F_a\left(\hat{X}\right) \le F_a\left(\hat{\omega}\right)\;,
\end{equation}
where $\hat{\omega}$ is the thermal Gaussian state such that $S_a(\hat{\omega}) = S_a\left(\hat{X}\right)$.
\begin{proof}
If $\hat{X}$ has rank $1$, we have from Lemma \ref{lem:pas}
\begin{equation}
F_a\left(\hat{X}\right) \le F_a\left(\hat{X}^\downarrow\right) = F_a\left(|0\rangle\langle 0|\right)\;,
\end{equation}
and the claim follows since $S_a\left(\hat{X}\right)=0$.
If $\hat{X}$ has rank at least $2$, then $S_a\left(\hat{X}\right)>0$ and the claim follows from Theorem \ref{thm:logs} with $0<z<1$ chosen such that $S_a(\hat{\omega}_z) = S_a\left(\hat{X}\right)$.
Such $z$ always exists since from \eqref{eq:omegaz}
\begin{equation}
\frac{\hat{\omega}_z^\frac{1}{a}}{\mathrm{Tr}\,\hat{\omega}_z^\frac{1}{a}} = \hat{\omega}_{z^\frac{1}{a}}\;.
\end{equation}
When $z$ covers the interval $(0,1)$, $z^\frac{1}{a}$ covers the interval $(0,1)$, too, and the claim follows since the entropy of a thermal Gaussian state can take any positive value.
\end{proof}
\end{cor}

\subsection{Proof of Theorem \ref{thm:logs}}
The starting point of the proof is the recent result of Ref. \cite{de2015passive}, that links the $p\to q$ norms to the notions of passive states.
The passive states of a quantum system \cite{pusz1978passive,lenard1978thermodynamical,gorecki1980passive} minimize the average energy for a given spectrum.
They are diagonal in the energy eigenbasis, and their eigenvalues decrease as the energy increases.

\begin{defn}[Fock rearrangement]
Let $\hat{X}$ be a positive operator with eigenvalues in decreasing order $x_0\ge x_1\ge\ldots\ge0$.
We define its Fock rearrangement as
\begin{equation}
\hat{X}^\downarrow:=\sum_{n=0}^\infty x_n\;|n\rangle\langle n|\;.
\end{equation}
If $\hat{X}$ coincides with its own Fock rearrangement, i.e.\ $\hat{X}=\hat{X}^\downarrow$, we say that it is {\it passive} \cite{pusz1978passive,lenard1978thermodynamical,gorecki1980passive}.
\end{defn}
\begin{rem}
From \eqref{eq:omegaz}, any thermal Gaussian state is passive.
\end{rem}
\begin{rem}
We have $\hat{X}^\downarrow = \hat{U}\,\hat{X}\,{\hat{U}}^\dag$, where $\hat{U}$ is the unitary operator that for any $n\in\mathbb{N}$ sends the eigenvector of $\hat{X}$ with eigenvalue $x_n$ to the $n$-th Fock state $|n\rangle$.
\end{rem}

The result is the following:
\begin{thm}\label{thm:maj}
For any $p,q\ge1$ the Fock rearrangement of the input does not decrease the $q$ norm of the output, i.e.\ for any operator $\hat{X}$ with $\left\|\hat{X}\right\|_p<\infty$ and any $0\leq\lambda\leq1$
\begin{equation}
\left\|\mathcal{E}_\lambda\left(\hat{X}\right)\right\|_q\leq\left\|\mathcal{E}_\lambda\left(\hat{X}^\downarrow\right)\right\|_q\;.
\end{equation}
\begin{proof}
From \cite[Theorem 34]{de2015passive}, $\mathcal{E}_\lambda\left(\hat{X}^\downarrow\right)$ majorizes $\mathcal{E}_\lambda\left(\hat{X}\right)$.
Then, from \cite[Theorem 19]{de2015passive} there exists a probability measure $\mu$ on the set of unitary operators such that
\begin{equation}
\mathcal{E}_\lambda\left(\hat{X}\right) = \int\hat{U}\,\mathcal{E}_\lambda\left(\hat{X}^\downarrow\right)\,{\hat{U}}^\dag\,\mathrm{d}p\left(\hat{U}\right)\;.
\end{equation}
Therefore,
\begin{equation}
\left\|\mathcal{E}_\lambda\left(\hat{X}\right)\right\|_q \le \int\left\|\hat{U}\,\mathcal{E}_\lambda\left(\hat{X}^\downarrow\right)\,{\hat{U}}^\dag\right\|_q\mathrm{d}p\left(\hat{U}\right) =\left\|\mathcal{E}_\lambda\left(\hat{X}^\downarrow\right)\right\|_q\;.
\end{equation}
\end{proof}
\end{thm}

\begin{lem}\label{lem:Tr}
$F_a\left(\hat{X}\right)$ and $\mathcal{F}_{z,a}\left(\hat{X}\right)$ are invariant upon rescaling $\hat{X}$ by a positive constant.
\end{lem}
\begin{proof}
For any $\lambda>0$
\begin{align}
F_a\left(\lambda\,\hat{X}\right) &= \left.\frac{d}{dt}\ln\left\|\mathrm{e}^{t\mathcal{L}}\left(\lambda\,\hat{X}\right)\right\|_\frac{1}{a}\right|_{t=0} = \left.\frac{d}{dt}\left(\ln\lambda + \ln\left\|\mathrm{e}^{t\mathcal{L}}\left(\hat{X}\right)\right\|_\frac{1}{a}\right)\right|_{t=0}\nonumber\\
& = \left.\frac{d}{dt}\ln\left\|\mathrm{e}^{t\mathcal{L}}\left(\hat{X}\right)\right\|_\frac{1}{a}\right|_{t=0} = F_a\left(\hat{X}\right)\;,
\end{align}
hence $F_a\left(\hat{X}\right)$ is invariant.
Since $S_a\left(\lambda\,\hat{X}\right) = S_a\left(\hat{X}\right)$, $\mathcal{F}_{z,a}\left(\hat{X}\right)$ is invariant, too.
\end{proof}

\begin{lem}\label{lem:pas}
For any positive operator $\hat{X}$ with finite rank
\begin{equation}
F_a\left(\hat{X}\right) \le F_a\left(\hat{X}^\downarrow\right)\;,\qquad
\mathcal{F}_{z,a}\left(\hat{X}\right) \le \mathcal{F}_{z,a}\left(\hat{X}^\downarrow\right)
\end{equation}
\end{lem}
\begin{proof}
From Theorem \ref{thm:maj},
\begin{equation}\label{passiveS}
\ln\norm{\mathrm{e}^{t\mathcal{L}}\left(\hat{X}\right)}_{1/a} \le \ln\norm{\mathrm{e}^{t\mathcal{L}}\left(\hat{X}^{\downarrow}\right)}_{1/a}\;.
\end{equation}
Since $\hat{X}$ and $\hat{X}^{\downarrow}$ have the same spectrum, the two sides of \eqref{passiveS} coincide at $t=0$, hence
\begin{equation}
F_a\left(\hat{X}\right) = \left.\frac{d}{dt}\ln\norm{\mathrm{e}^{t\mathcal{L}}\left(\hat{X}\right)}_{1/a}\right|_{t=0} \le \left.\frac{d}{dt}\ln\norm{\mathrm{e}^{t\mathcal{L}}\left(\hat{X}^{\downarrow}\right)}_{1/a}\right|_{t=0} = F_a\left(\hat{X}^\downarrow\right)\;.
\end{equation}
Moreover, $S_a\left(\hat{X}\right) = S_a\left(\hat{X}^\downarrow\right)$, and the claim follows.
\end{proof}

With some abuse of notation, we define for any $x\in\mathbb{R}^{N+1}$
\begin{align}\label{eq:defFa}
F_a(x) &:= \sum_{n=1}^N n\,\left(x_n\,|x_{n-1}|^{\frac{1}{a}-1} - |x_n|^{\frac{1}{a}}  \right)\;,\\
S_a(x) &:= -\sum_{n=0}^N |x_n|^\frac{1}{a}\ln|x_n|^\frac{1}{a}\;,\\
\mathcal{F}_{z,a}(x) &:= F_a(x)+\mu(z,a)\,S_a(x)\;.
\end{align}
This choice is motivated by the following Lemma.
\begin{lem}\label{lem:dlX}
For any $x\in\mathbb{R}^{N+1}$ with nonnegative components such that
\begin{equation}\label{eq:Trx}
\sum_{n=0}^N |x_n|^\frac{1}{a}=1\;,
\end{equation}
let
\begin{equation}\label{eq:defX}
\hat{X}=\sum_{n=0}^N x_n\;|n\rangle\langle n|\;.
\end{equation}
Then,
\begin{equation}
F_a(x) = F_a\left(\hat{X}\right)\;,\qquad S_a(x) = S_a\left(\hat{X}\right)\;,\qquad \mathcal{F}_{z,a}(x) = \mathcal{F}_{z,a}\left(\hat{X}\right)\;.
\end{equation}
\begin{proof}
From \cite[Eqs. (VI.3), (VI.4)]{de2016gaussian}, we have
\begin{equation}
\mathrm{e}^{t\mathcal{L}}\left(\hat{X}\right)=\sum_{n=0}^N x_n(t)\;|n\rangle\langle n|\;,
\end{equation}
where for $n=0,\ldots,\,N$
\begin{equation}
x_n(t)=\sum_{k=n}^N\binom{k}{n}\mathrm{e}^{-nt}\left(1-\mathrm{e}^{-t}\right)^{k-n}x_k\;.
\end{equation}
The claim then follows from an explicit computation and from condition \eqref{eq:Trx}.
\end{proof}
\end{lem}
In the remaining part of this section, we will prove the following Theorem.
\begin{thm}\label{thm:logsx}
For any $N\in\mathbb{N}$ and any $x\in\mathbb{R}^{N+1}$ with nonnegative components and satisfying \eqref{eq:Trx},
\begin{equation}
\mathcal{F}_{z,a}(x) \le \mathcal{F}_{z,a}\left(\hat{\omega}_z\right)\;,\qquad F_a(x)\le 1-a\;.
\end{equation}
\end{thm}
\begin{prop}
Theorem \ref{thm:logsx} implies Theorem \ref{thm:logs}.
\end{prop}
\begin{proof}
Lemma \ref{lem:dlX} and Theorem \ref{thm:logsx} imply that \eqref{logs} and \eqref{eq:Fa} hold for any positive operator $\hat{X}$ with finite rank, diagonal in the Fock basis and with $\mathrm{Tr}\,\hat{X}^{1/a}=1$.
From Lemma \ref{lem:pas}, \eqref{logs} and \eqref{eq:Fa} hold for any positive operator $\hat{X}$ with finite rank and with $\mathrm{Tr}\,\hat{X}^{1/a}=1$, and the claim of Theorem \ref{thm:logs} follows from Lemma \ref{lem:Tr}.
\end{proof}

\begin{lem}\label{lem:bounded}
For any $x\in\mathbb{R}^{N+1}$ satisfying \eqref{eq:Trx},
\begin{equation}\label{eq:bounded}
F_a(x)\le 1-a\;.
\end{equation}
\begin{proof}
We define for any $n=0,\,\ldots,\,N$
\begin{equation}
p_n:=x_n^\frac{1}{a}\;,\qquad I:=\{n\in\{0,\,\ldots,\,N\}:p_n>0\}\;,\qquad  E:=\sum_{n\in I} n\,p_n\;,
\end{equation}
such that
\begin{equation}\label{eq:Fp}
F_a(x) = \sum_{n\in I} n\,p_n\left(\left(\frac{p_{n-1}}{p_n}\right)^{1-a} - 1\right)\;.
\end{equation}
From the hypothesis \eqref{eq:Trx}, $p$ is a probability measure on $I$.
If $E=0$, then $p_0=1$ and $p_1=\ldots=p_N=0$, and the right-hand side of \eqref{eq:Fp} vanishes.
We can then suppose $E>0$.
Since the function $s\mapsto s^{1-a}-1$ is concave, Jensen's inequality implies
\begin{align}
E\sum_{n\in I} \frac{np_n}{E}\left(\left(\frac{p_{n-1}}{p_n}\right)^{1-a}-1\right) &\le E\left(\left(\sum_{n\in I}\frac{np_n}{E}\frac{p_{n-1}}{p_n}\right)^{1-a}-1\right)\nonumber\\
&=E\left(\left(\frac{1}{E}\sum_{n=1}^N n\,p_{n-1}\right)^{1-a}-1\right)\nonumber\\
&=E\left(\left(\frac{E+1-\left(N+1\right)p_N}{E}\right)^{1-a}-1\right)\nonumber\\
&\le E\left(\left(1+\frac{1}{E}\right)^{1-a}-1\right)\le1-a\;,
\end{align}
where in the last step we have used the inequality
\begin{equation}
(1+s)^{1-a} \le 1+\left(1-a\right)s\;,\qquad s\ge0\;.
\end{equation}
\end{proof}
\end{lem}

Let $\mathcal{P}_N\subset \mathbb{R}^{N+1}$ be the set of vectors with nonnegative components satisfying \eqref{eq:Trx}.
Since $0<a<1$, $\mathcal{F}_{z,a}$ is continuous on $\mathbb{R}^{N+1}$ and $\mathcal{P}_N$ is compact.
Then, the restriction of $\mathcal{F}_{z,a}$ to $\mathcal{P}_N$ admits a global maximizer $\bar{x}\in\mathcal{P}_N$.

\begin{lem}\label{lem:x*}
The restriction of $\mathcal{F}_{z,a}$ to $\mathcal{P}_N$ always admits a global maximizer $\bar{x}\in\mathcal{P}_N$ such that
\begin{equation}
\bar{x}_0\ge\ldots\ge \bar{x}_N\ge0\;.
\end{equation}
\begin{proof}
Let $\bar{x}\in\mathcal{P}_N$ be a global maximizer of the restriction of $\mathcal{F}_{z,a}$ to $\mathcal{P}_N$.
Lemma \ref{lem:pas} and Lemma \ref{lem:dlX} imply that
\begin{equation}
\mathcal{F}_{z,a}\left(\bar{x}\right) \le \mathcal{F}_{z,a}\left(\bar{x}^{\downarrow}\right)\;,
\end{equation}
where $\bar{x}^{\downarrow}$ is the decreasing rearrangement of $\bar{x}$, i.e., $\bar{x}^{\downarrow}_n = \bar{x}_{\sigma(n)}$ for any $n=0,\,\ldots,\,N$, where $\sigma\in S_{N+1}$ is a permutation such that
\begin{equation}
\bar{x}_{\sigma(0)} \ge\ldots \ge \bar{x}_{\sigma(N)}\ge0\;.
\end{equation}
Since also $\bar{x}^{\downarrow}\in\mathcal{P}_N$, $\bar{x}^{\downarrow}$ is a global maximizer of the restriction of $\mathcal{F}_{z,a}$ to $\mathcal{P}_N$, too, and the claim follows.
\end{proof}
\end{lem}
From now on, we assume that $\bar{x}=\bar{x}^{\downarrow}$, i.e., $\bar{x}_0\ge \ldots \ge \bar{x}_N\ge0$.
Let $N'$ be such that
\begin{equation}\label{eq:defN'}
\bar{x}_0\ge\ldots\ge \bar{x}_{N'}>0\;,\qquad \bar{x}_{N'+1}=\ldots= \bar{x}_N= 0\;.
\end{equation}
If $\bar{x}_0\ge \ldots \ge \bar{x}_N>0$, we set $N'=N$.
Even if $N'$ in principle depends on the choice of $\bar{x}$, we omit this dependence for the sake of a simpler notation.
\begin{lem}\label{lem:reg}
The hypotheses of the Karush--Kuhn--Tucker conditions (Theorem \ref{thm:KKT}) for the maximization of $\mathcal{F}_{z,a}:\mathbb{R}^{N+1}\to\mathbb{R}$ with the constraint functions
\begin{equation}
\psi_n(x) = x_n\;,\quad n=0,\,\ldots,\,N,\qquad \phi(x) = \sum_{n=0}^N |x_n|^\frac{1}{a} - 1
\end{equation}
are satisfied in $\bar{x}$.
Therefore, $\bar{x}$ satisfies the stationarity condition \eqref{eq:KKT}, i.e., there exists $\lambda\in\mathbb{R}$ such that
\begin{align}\label{eq:KKTimpl1}
\frac{\partial\mathcal{F}_{z,a}}{\partial x_n}(\bar{x}) &= \lambda\frac{\partial \phi}{\partial x_n}(\bar{x})\;,\qquad n=0,\,\ldots,\,N'\;,\\
\frac{\partial\mathcal{F}_{z,a}}{\partial x_n}(\bar{x}) &\le \lambda\frac{\partial \phi}{\partial x_n}(\bar{x})\;,\qquad n=N'+1,\,\ldots,\,N\;,\label{eq:KKTimpl2}
\end{align}
where $N'$ is as in \eqref{eq:defN'}.
An explicit computation of \eqref{eq:KKTimpl1} yields
\begin{align}\label{eq:KKTst}
&n\,\bar{x}_{n-1}^{\frac{1-a}{a}}-\frac{n}{a}\,\bar{x}_{n}^{\frac{1-a}{a}} + \left(n+1\right)\frac{1-a}{a}\,\bar{x}_{n+1}\,\bar{x}_{n}^{\frac{1-2a}{a}} - \frac{\mu(z,a)}{a}\left(\ln \bar{x}_n^{\frac{1}{a}}+1\right) \bar{x}_{n}^{\frac{1-a}{a}}\nonumber\\
&=\frac{\lambda}{a}\,\bar{x}_{n}^{\frac{1-a}{a}}\;,\qquad n=0,\,\ldots,\,N'\;,
\end{align}
while \eqref{eq:KKTimpl2} for $n=N'+1$ yields
\begin{equation}\label{eq:KKTdf}
\left(N'+1\right)\bar{x}_{N'}^{\frac{1-a}{a}}\le0\;.
\end{equation}
\end{lem}
\begin{proof}
Let us prove that the hypotheses of Theorem \ref{thm:KKT} are satisfied.
\begin{description}
\item[(a)] Since $0<a<1$, both $F_a$ and $\phi$ are continuous on $\mathbb{R}^{N+1}$, hence also $\mathcal{F}_{z,a}$ is continuous.
\item[(b)] The functions $\phi$ and $S_a$ are continuously differentiable on $\mathbb{R}^{N+1}$.
Let us prove that $F_a$ is differentiable in $\bar{x}$.
If $N'=N$ or $N'=N-1$, $F_a$ is continuously differentiable in a neighbourhood of $\bar{x}$.
Let us then assume $N'\le N-2$.
We recall that $\bar{x}_0\ge\ldots\ge \bar{x}_{N'}>0 = \bar{x}_{N'+1} = \ldots = \bar{x}_N$.
The only terms in the sum \eqref{eq:defFa} that are not continuously differentiable in a neighbourhood of $\bar{x}$ are
\begin{equation}\label{eq:sumnd}
\sum_{n=N'+2}^N n\,x_n\left|x_{n-1}\right|^{\frac{1-a}{a}}\;.
\end{equation}
From Lemma \ref{lem:diff}, the sum \eqref{eq:sumnd} is differentiable in $\bar{x}$, and the claim follows.
\item[(c)] Since $\phi(\bar{x})=0$ and $\bar{x}$ is nonnegative and decreasing, we have $\bar{x}_0>0$ and
\begin{equation}
\frac{\partial \phi}{\partial x_0}(\bar{x}) = \frac{\left(\bar{x}_0\right)^{\frac{1}{a}-1}}{a}>0\;,
\end{equation}
hence $\nabla \phi(\bar{x})\neq 0$.
\item[(d)] We have $I=\{N'+1,\,\ldots,\,N\}$.
If $N'=N$, $I$ is empty and the condition reduces to $\nabla \phi(\bar{x})\neq0$.
Let us then suppose $N'<N$.
We have for any $m=0,\,\ldots,\,N$ and any $n\in I$
\begin{equation}
\frac{\partial \psi_n}{\partial x_m}(\bar{x}) = \delta_{mn}\;,\qquad \frac{\partial \phi}{\partial x_m}(\bar{x}) = \frac{\left(\bar{x}_m\right)^{\frac{1}{a}-1}}{a}\;.
\end{equation}
Since $\bar{x}_m=0$ iff $m\in I$, the gradients $\nabla \phi(\bar{x})$ and $\{\nabla \psi_n(\bar{x})\}_{n\in I}$ are independent.
\end{description}
The conditions \eqref{eq:KKTimpl1} and \eqref{eq:KKTimpl2} follows from \eqref{eq:KKT} using that $\psi_n(x)=x_n$.
The condition \eqref{eq:KKTst} follows computing explicitly the derivatives in \eqref{eq:KKTimpl1}.
Since $\bar{x}_{N'+1}=0$ and $0<a<1$ we have
\begin{align}
\frac{\partial \phi}{\partial x_{N'+1}}(\bar{x}) &= \frac{\bar{x}_{N'+1}^{\frac{1-a}{a}}}{a} = 0\;,\nonumber\\
\frac{\partial S_a}{\partial x_{N'+1}}(\bar{x}) &= \left.-\frac{x_{N'+1}^\frac{1-a}{a}}{a}\left(\ln x_{N'+1}^{\frac{1}{a}} + 1\right)\right|_{x_{N'+1}=0} = 0\;.
\end{align}
Moreover,
\begin{align}
\frac{\partial F_a}{\partial x_{N'+1}}(\bar{x}) &= \left.\frac{\partial}{\partial x_{N'+1}}\left(\left(N'+1\right)x_{N'+1}\,\bar{x}_{N'}^{\frac{1-a}{a}} - x_{N'+1}^\frac{1}{a}\right)\right|_{x_{N'+1} = 0}\nonumber\\
&= \left(N'+1\right)\bar{x}_{N'}^{\frac{1-a}{a}}\;,
\end{align}
hence \eqref{eq:KKTimpl2} for $n=N'+1$ becomes \eqref{eq:KKTdf}.
\end{proof}

\begin{lem}
Let $\bar{x}\in\mathcal{P}_N$ be a global maximizer of the restriction of $\mathcal{F}_{z,a}$ to $\mathcal{P}_N$ such that $\bar{x}_0\ge \ldots \ge \bar{x}_N\ge0$.
Then, $\bar{x}_0\ge\ldots\ge \bar{x}_N>0$, i.e., its $N'$ defined as in \eqref{eq:defN'} is equal to $N$.
\begin{proof}
If $N'<N$, condition \eqref{eq:KKTdf} is in contradiction with $\bar{x}_{N'} >0$.
\end{proof}
\end{lem}

The following change of variables simplifies the KKT stationarity condition \eqref{eq:KKTst}:
\begin{equation}\label{eq:p}
p_n := \bar{x}_n^{\frac{1}{a}}\;,\quad n = 0,\,\ldots,\,N\;,\qquad p_{-1}:=1\;,\qquad p_{N+1}:=0\;.
\end{equation}
From the constraint \eqref{eq:Trx}, $p$ is a probability measure on $\{0,\ldots,\,N\}$.
We also define for any $0<s<1$
\begin{equation}
\xi := z^\frac{1}{a}\;,\quad \nu(s,a) := \mu(s^a,a) = \frac{a\,s^{a-1}-1+\left(1-a\right)s^a}{\ln s}\;,\quad h(s) := s^{a-1} -1\;.
\end{equation}
Bringing all the terms to the left-hand side, the KKT stationarity condition \eqref{eq:KKTst} reads
\begin{equation}\label{eq:condK}
\frac{p_n^{1-a}}{a}\,\mathcal{K}_n = 0\;,\qquad n=0,\,\ldots,\,N\;,
\end{equation}
where
\begin{align}\label{eq:el-equation}
\mathcal{K}_n &= n\,h\bra{\frac{p_n}{p_{n-1}}} + \frac{n\,p_n}{p_{n-1}}\;h'\bra{ \frac{p_n}{p_{n-1}}} -\left(n+1\right) \bra{\frac{p_{n+1}}{p_{n}}}^2 h'\bra{ \frac{p_{n+1}}{p_{n}}}\nonumber\\
&\phantom{=}-\nu(\xi,a) \left(\ln p_n +1\right) - \lambda\;,\qquad n=0,\,\ldots,\,N\;,
\end{align}
and the function $s\mapsto s^2\,h'(s)$ is set to $0$ in $s=0$ by continuity.

For $n = 0, \dots, N$ we define
\begin{equation}\label{eq:w}
w_n := \frac{p_{n+1}}{p_n}\;,
\end{equation}
so that $0< w_n \le 1$ for $n=0,\ldots,\,N-1$ and $w_N = 0$.
For $n =1, \ldots, N$ we have from \eqref{eq:condK}
\begin{align}\label{eq:ric0}
0&=\mathcal{K}_n - \mathcal{K}_{n-1}\nonumber\\
&=   n\,h(w_{n-1}) + n\,w_{n-1}\,h'(w_{n-1}) -\left(n+1\right) w_n^2\,h'(w_n) - \bra{n-1} h(w_{n-2})\nonumber\\
&\phantom{=} - \bra{n-1} w_{n-2}\,h'(w_{n-2}) + n w_{n-1}^2\,h'(w_{n-1})  - \nu(w,a) \ln w_{n-1}\;.
\end{align}
We notice that
\begin{equation}
\nu(s,a)=\frac{h(s)+s\,h'(s)-s^2\,h'(s)}{\ln s}\;.
\end{equation}
Then, \eqref{eq:ric0} can be recast as
\begin{align}\label{eq:z_n-equation}
&\bra{ n+1} \left(w_n^2\,h'(w_n) - w_{n-1}^2\,h'(w_{n-1})\right)\nonumber\\
&= \bra{n-1} \left(h(w_{n-1}) + w_{n-1}\,h'(w_{n-1}) - h(w_{n-2}) - w_{n-2}\,h'(w_{n-2})\right)\nonumber\\
&\phantom{=}+  \left(\nu(w_{n-1},a)- \nu(w,a)\right) \ln w_{n-1}\;.
\end{align}

\begin{lem}\label{lem:nu}
For any $0<a<1$ and any $0<s<1$, the function $s\mapsto\nu(s,a)$ is strictly increasing and maps $(0,1)$ into $(-\infty,0)$.
\begin{proof}
We have to prove
\begin{equation}
 0 \overset{?}{<}\frac{\partial}{\partial s} \nu(s,a) = \frac{a(a-1)s^{a-1}(1-s) \ln s - as^{a-1}+(a-1)s^a +1}{s (\ln s )^2}\;,
\end{equation}
which reduces to the inequality
\begin{equation}\label{eq:dnu}
-a\left(1-a\right)\left(1-s\right) \ln s -a - \left(1-a\right)s + s^{1-a} \overset{?}{>} 0\;.
\end{equation}
Since the function $s\mapsto s^{1-a}$ is concave, we have for any $0<s<1$
\begin{equation}\label{eq:conca}
s^{1-a} > 1 - \left(1-a\right)\frac{1-s}{s^a}\;.
\end{equation}
Moreover,
\begin{equation}\label{eq:concln}
-a\ln s = \ln\frac{1}{s^a} > \frac{1}{s^a}-1\;,
\end{equation}
and \eqref{eq:dnu} follows since
\begin{align}
&-a\left(1-a\right)\left(1-s\right) \ln s -a - \left(1-a\right)s + s^{1-a}\nonumber\\
&>\left(1-a\right)\left(1-s\right)\left(\frac{1}{s^a}-1\right)-a-\left(1-a\right)s+1- \left(1-a\right)\frac{1-s}{s^a} = 0\;.
\end{align}

For the second part, it is sufficient to notice that $\lim_{s\to0}\nu(s,a)=-\infty$ and $\lim_{s\to1}\nu(1,a)=0$.
\end{proof}
\end{lem}

\begin{lem}\label{lem:zdecr}
Let $\bar{x}\in\mathcal{P}_N$ be a global maximizer of the restriction of $\mathcal{F}_{z,a}$ to $\mathcal{P}_N$ such that $\bar{x}_0\ge \ldots \ge \bar{x}_N\ge0$, and let $w\in\mathbb{R}^{N+1}$ be as in \eqref{eq:w}.
Then, $z^\frac{1}{a} = \xi \ge w_0 \ge w_1 \ge\ldots \ge w_N = 0$.
\begin{proof}
Let us first assume that $\xi<w_0$, hence $\nu(\xi,a)<\nu(w_0,a)$ from Lemma \ref{lem:nu}.
We then have from \eqref{eq:z_n-equation} with $n =1$
\begin{equation}
2 \left( w_1^2\,h'(w_1) - w_{0}^2\,h'(w_{0}) \right) = \left(\nu(w_0,a)  - \nu(w,a) \right)\ln w_{0} <0\;,
\end{equation}
thus $w_1 > w_0$ since the function
\begin{equation}
s\mapsto s^2\,h'(s) = -\left(1-a\right)s^a
\end{equation}
is strictly decreasing.

We will prove by induction that $w_n > w_{n-1}$ for $n =1, \ldots, N$.
This yields a contradiction for $n = N$ since $w_N = 0$.
The claim is true for $n=1$.
We can assume from the inductive hypothesis $w_{n-1}>w_{n-2}>\ldots>w_0$, $n\ge2$.
We then have $\nu(w_{n-1},a)>\nu(w_0,a)>\nu(\xi,a)$.
Since the function
\begin{equation}
s\mapsto h(s) + s\,h'(s) = \frac{a}{s^{1-a}} - 1
\end{equation}
is decreasing, we get from \eqref{eq:z_n-equation}
\begin{align}
&\bra{ n+1} \left( w_n^2\,h'(w_n) - w_{n-1}^2\,h'(w_{n-1})\right)\nonumber\\
&= \bra{n-1} \left(h(w_{n-1}) + w_{n-1}\,h'(w_{n-1}) - h(w_{n-2}) - w_{n-2}\,h'(w_{n-2})\right)  \nonumber\\
&\phantom{=} +\left(\nu(w_{n-1},a)  - \nu(w,a)\right) \ln w_{n-1}\nonumber\\
&<  \bra{n-1} \left( h(w_{n-1}) + w_{n-1}\,h'(w_{n-1}) - h(w_{n-2}) - w_{n-2}\,h'(w_{n-2})\right)<0\;.
\end{align}
Since the function $s\mapsto s^2 h'(s)$ is strictly decreasing, we have $w_{n} > w_{n-1}$.

We must then have $w_0<\xi$, hence $\nu\left(w_0,a\right) \le \nu(\xi,a)$.
Thus, from \eqref{eq:z_n-equation} with $n =1$,
\begin{equation}
2 \left( w_1^2\,h'(w_1) - w_{0}^2\,h'(w_{0})\right) = \left(\nu(w_0,a)  - \nu(\xi,a)\right) \ln w_{0} \ge 0\;,
\end{equation}
thus $w_1 \le w_0$ since the function $s\mapsto s^2 h'(s)$ is strictly decreasing.

Let us prove by induction that $w_n \le w_{n-1}$ for $n =1, \ldots, N$.
The claim is true for $n=1$.
We can assume from the inductive hypothesis $w_{n-1}\le w_{n-2}\le\ldots\le w_0$, $n\ge2$, hence $\nu\left(w_{n-1},a\right)\le\nu(w_0,a)\le\nu(\xi,a)$.
Since the function $s\mapsto h(s) + s h'(s)$ is decreasing, we get from \eqref{eq:z_n-equation}
\begin{align}
&\bra{ n+1} \left( w_n^2\,h'(w_n) - w_{n-1}^2\,h'(w_{n-1})\right)\nonumber\\
& = \bra{n-1} \left(h(w_{n-1}) + w_{n-1}\,h'(w_{n-1}) - h(w_{n-2}) - w_{n-2}\,h'(w_{n-2})\right) \nonumber\\
&\phantom{=}+\left(\nu(w_{n-1},a)  - \nu(\xi,a)\right) \ln w_{n-1}\nonumber\\
 & \ge  \bra{n-1} \left(h(w_{n-1}) + w_{n-1}\,h'(w_{n-1}) - h(w_{n-2}) - w_{n-2}\,h'(w_{n-2})\right) \ge 0\;,
\end{align}
and since the function $s\mapsto s^2 h'(s)$ is strictly decreasing, we have $w_{n} \le  w_{n-1}$.
\end{proof}
\end{lem}

We now consider the dependence on $N$.
For any $N\in\mathbb{N}$, let $\bar{x}^{(N)}$ be a maximizer of the restriction $\mathcal{F}_{z,a}$ to $\mathcal{P}_N$ such that $\bar{x}^{(N)}_0\ge\ldots\ge \bar{x}^{(N)}\ge0$, and let $p^{(N)}$ and $w^{(N)}\in\mathbb{R}^{N+1}$ be as in Eqs. \eqref{eq:p} and \eqref{eq:w}, respectively.
With a diagonal argument, there exists a subsequence $\left\{N_k\right\}_{k\in\mathbb{N}}$ such that for any $n \in \mathbb{N}$
\begin{equation}
\lim_{k \to \infty} w^{(N_k)}_n = w^\infty_n\;,\qquad\lim_{k\to\infty}p^{(N_k)}_n=p_n^\infty\;.
\end{equation}
Since monotonicity is preserved in the limit, we have
\begin{equation}\label{eq:monz}
1>\xi\ge w^\infty_0\ge w^\infty_1\ge\ldots\ge0\;.
\end{equation}

\begin{lem}\label{lem:znric}
The $w^\infty_n$ are all strictly positive and satisfy for any $n\ge1$ the recursive relation
\begin{align}\label{eq:z-equation}
&\bra{ n+1} \sqa{ \left(w^\infty_n\right)^2 h'(w^\infty_n) - \left(w^{\infty}_{n-1}\right)^2 h'(w^\infty_{n-1}) }\nonumber\\
&= \bra{n-1} \left( h\bra{w^\infty_{n-1}} + w^\infty_{n-1}\,h'\bra{ w^\infty_{n-1} } - h\bra{w^\infty_{n-2}} - w^\infty_{n-2}\,h'\bra{ w^\infty_{n-2}}\right)\nonumber\\
&\phantom{=} +\left(\nu(w^\infty_{n-1},a)  - \nu(w,a)\right) \ln w^\infty_{n-1}\;.
\end{align}
\begin{proof}
If $w^\infty_0=0$, since the sequence $n\mapsto w_n^{(N)}$ is decreasing we have for any $n\in\mathbb{N}$
\begin{equation}
\limsup_{k\to\infty}w_n^{(N_k)}\le\limsup_{k\to\infty}w_0^{(N_k)}=w^\infty_0=0\;,
\end{equation}
hence $w^\infty_n=\lim_{k\to\infty}w_n^{(N_k)}=0$ for any $n$.
Since $w^{(N)}_{n-1}\le w^{(N)}_{n-2}$ and the function $s\mapsto h(s)+sh'(s)$ is decreasing, we get from \eqref{eq:z_n-equation}
\begin{equation}\label{eq:mu0}
\nu(\xi,a)\le\nu\left(w_{n-1}^{(N_k)},a\right)+\left(n+1\right)\frac{\left(w_n^{(N_k)}\right)^2 h'\left(w_n^{(N_k)}\right)-\left(w_{n-1}^{(N_k)}\right)^2 h'\left(w_{n-1}^{(N_k)}\right)}{-\ln w^{(N_k)}_{n-1}}.
\end{equation}
Since $\lim_{s\to0}s^2h'(s)=0$ and $\lim_{s\to0}\nu(s,a)=-\infty$, the right-hand side of \eqref{eq:mu0} tends to $-\infty$ for $k\to\infty$, giving a contradiction.

We must then have $w^\infty_0>0$.
We proceed by induction on $n$.
From the inductive hypothesis, we can suppose
\begin{equation}
w^\infty_0=\lim_{k\to\infty}w_0^{(N_k)}\ge\ldots\ge\lim_{k\to\infty}w_{n-1}^{(N_k)}=w^\infty_{n-1}>0\;.
\end{equation}
Let us suppose $w^\infty_n=0$.
Since the function $s\mapsto s^2h'(s)$ is nonpositive and continuous for any $0\le s\le1$, taking the limit $k\to\infty$ in \eqref{eq:z_n-equation} with $n+1$ in place of $n$ we get
\begin{align}\label{eq:zn}
0 &\ge \left(n+2\right)\left(w_{n+1}^{\infty}\right)^2 h'(w_{n+1}^\infty)\nonumber\\
&= \lim_{k\to\infty}\left(\left(n+1\right)\left(h\left(w_n^{(N_k)}\right)+w_n^{(N_k)}\;h'\left(w_n^{(N_k)}\right)\right)\nonumber-\nu(\xi,a)\ln w_n^{(N_k)}\right)\nonumber\\
&\phantom{=}-h(w_{n-1}^\infty)-w_{n-1}^\infty\;h'(w_{n-1}^\infty)=\infty\;,
\end{align}
where we have used that $h(s)+sh'(s)=\mathcal{O}(s^{a-1})$ and $s^2h'(s)\to0$ for $s\to0$.
Since \eqref{eq:zn} contains a contradiction, we must have $w_n^\infty>0$, and \eqref{eq:z-equation} follows taking the limit $k\to\infty$ in \eqref{eq:z_n-equation}.
\end{proof}
\end{lem}
\begin{lem}\label{lem:limzn}
$w_n^\infty=\xi<1$ for any $n\in\mathbb{N}$.
\begin{proof}
From \eqref{eq:monz} and Lemma \ref{lem:znric} we have $0<w_0^\infty\le \xi<1$.
Since the sequence $n\mapsto w_n^\infty$ is positive and decreasing, it has a limit $\lim_{n\to\infty}w_n^\infty=\inf_{n\in\mathbb{N}}w_n^\infty=\xi'$, that satisfies
\begin{equation}\label{eq:inzbar}
0\le \xi'\le w_0^\infty\le \xi<1\;.
\end{equation}
Since the function $s\mapsto h(s)+sh'(s)$ is decreasing and $w_{n-1}^\infty\le w_{n-2}^\infty$, the recursive relation \eqref{eq:z-equation} implies
\begin{align}
&\left(n+1\right)\left(\left(w_n^{\infty}\right)^2 h'(w_n^\infty)-\left(w_{n-1}^{\infty}\right)^2 h'(w_{n-1}^\infty)\right)\nonumber\\
&\ge (\nu(w_{n-1}^\infty,a)-\nu(\xi,a))\ln w_{n-1}^\infty\;,
\end{align}
hence
\begin{equation}\label{eq:mul}
\nu(\xi,a)\le\frac{\left(n+1\right)\left(\left(w_n^{\infty}\right)^2 h'(w_n^\infty)-\left(w_{n-1}^{\infty}\right)^2 h'(w_{n-1}^\infty)\right)}{-\ln w_{n-1}^\infty}+\nu(w_{n-1}^\infty,a)\;.
\end{equation}
Since the function $s\mapsto s^2h'(s)$ is bounded and decreasing for $0\le s\le1$, we have
\begin{equation}
\sum_{n=1}^\infty\left(\left(w_n^{\infty}\right)^2 h'(w_n^\infty)-\left(w_{n-1}^{\infty}\right)^2 h'(w_{n-1}^\infty)\right)={\xi'}^2\,h'(\xi')-\left(w_0^{\infty}\right)^2 h'(w_0^\infty)<\infty\;.
\end{equation}
Since $w_n^\infty\le w_{n-1}^\infty$, the sequence $\left\{\left(w_n^{\infty}\right)^2 h'(w_n^\infty)-\left(w_{n-1}^{\infty}\right)^2 h'(w_{n-1}^\infty)\right\}_{n\in\mathbb{N}_0}$ is positive and summable, hence
\begin{equation}
\liminf_{n\to\infty}\left(n+1\right)\left(\left(w_n^{\infty}\right)^2 h'(w_n^\infty)-\left(w_{n-1}^{\infty}\right)^2 h'(w_{n-1}^\infty)\right)=0\;.
\end{equation}
Since $-\ln w_{n-1}^\infty\ge-\ln w_0^\infty>0$, we also have
\begin{equation}
\liminf_{n\to\infty}\frac{\left(n+1\right)\left(\left(w_n^{\infty}\right)^2 h'(w_n^\infty)-\left(w_{n-1}^{\infty}\right)^2 h'(w_{n-1}^\infty)\right)}{-\ln w_{n-1}^\infty}=0\;,
\end{equation}
and taking the $\liminf$ for $n\to\infty$ of \eqref{eq:mul} we get
\begin{equation}
\nu(\xi,a)\le\nu(\xi',a)\;,
\end{equation}
hence $\xi\le \xi'$.
Since from \eqref{eq:inzbar} also the converse inequality holds, we must have $\xi'=w_0^\infty=\xi$.
Since the sequence $n\mapsto w_n^\infty$ is decreasing and $\xi=\xi'=\inf_{n\in\mathbb{N}}w_n^\infty$, we have $w_0^\infty=\xi\le w_n^\infty\le w_0^\infty$ for any $n$, hence $w_n^\infty=\xi$.
\end{proof}
\end{lem}
\begin{lem}\label{lem:limqn}
$\lim_{k\to\infty}p_n^{(N_k)}=p_0^\infty\,\xi^n$ for any $n\in\mathbb{N}$.
\begin{proof}
The claim is true for $n=0$.
The inductive hypothesis is
\begin{equation}
\lim_{k\to\infty}p_{n'}^{(N_k)}=p_0^\infty \xi^{n'}\;,\qquad n'=0,\ldots,n\;.
\end{equation}
We then have $\lim_{k\to\infty}p^{(N_k)}_{n+1}=\lim_{k\to\infty}p_n^{(N_k)}w_n^{(N_k)}=p_0^\infty\,\xi^{n+1}$, where we have used the inductive hypothesis and Lemma \ref{lem:limzn}.
\end{proof}
\end{lem}
\begin{lem}\label{lem:limq0}
We have $p_0^\infty=1-\xi$, hence $\lim_{k\to\infty} p_n^{(N_k)}=\left(1-\xi\right)\xi^n$ for any $n\in\mathbb{N}$.
\begin{proof}
We have $\sum_{n=0}^N p_n^{(N)}=1$ for any $N\in\mathbb{N}$.
Moreover, since the sequence $n\mapsto w_n^{(N)}$ is decreasing, we also have
\begin{equation}
p_n^{(N)}=p_0^{(N)}w_0^{(N)}\ldots w_{n-1}^{(N)}\le p_0^{(N)}\left(w_0^{(N)}\right)^n\;.
\end{equation}
We have $\lim_{k\to\infty}w_0^{(N_k)}=\xi<1$, hence $w_0^{(N_k)}\le(1+\xi)/2$ for sufficiently large $k$, and since $p_0^{(N)}\le1$,
\begin{equation}\label{eq:boundqnk}
p_n^{(N_k)}\le\left(\frac{1+\xi}{2}\right)^n\;.
\end{equation}
The sums $\sum_{n=0}^{N_k}p_n^{(N_k)}$ are then dominated for any $k\in\mathbb{N}$ by $\sum_{n=0}^\infty\left(\frac{1+\xi}{2}\right)^n<\infty$, and from the dominated convergence theorem we have
\begin{equation}
1 = \lim_{k\to\infty}\sum_{n=0}^{N_k}p_n^{(N_k)}=\sum_{n=0}^\infty\lim_{k\to\infty}p_n^{(N_k)}=p_0^\infty\sum_{n=0}^\infty \xi^n =\frac{p_0^\infty}{1-\xi}\;,
\end{equation}
where we have used Lemma \ref{lem:limqn}.
\end{proof}
\end{lem}
We notice that for any $N\in\mathbb{N}$
\begin{equation}
S_a\left(\bar{x}^{(N)}\right) = S\left(p^{(N)}\right) = -\sum_{n=0}^N p_n^{(N)}\ln p_n^{(N)}\;.
\end{equation}
\begin{lem}
$\lim_{k\to\infty}S\left(p^{(N_k)}\right)=S(\hat{\omega}_\xi)$.
\begin{proof}
The function $x\mapsto -x\ln x$ is increasing for $0\le x\le1/\mathrm{e}$.
Let us choose $n_0$ such that $((1+\xi)/2)^{n_0}\le1/\mathrm{e}$.
Recalling \eqref{eq:boundqnk}, the sums
\begin{equation}
-\sum_{n=n_0}^{N_k}p_n^{(N_k)}\ln p_n^{(N_k)}
\end{equation}
are dominated for any $k\in\mathbb{N}$ by $-\sum_{n=n_0}^\infty n\left(\frac{1+\xi}{2}\right)^n\ln\frac{1+\xi}{2}<\infty$.
Then, from the dominated convergence theorem and Lemma \ref{lem:limq0} we have
\begin{align}
\lim_{k\to\infty}S\left(p^{(N_k)}\right) &= -\sum_{n=0}^\infty p_n^{(N_k)}\ln p_n^{(N_k)}\nonumber\\
&=-\sum_{n=0}^\infty\left(1-\xi\right)\xi^n\left(\ln\left(1-\xi\right)+n\ln \xi\right)=S(\hat{\omega}_\xi)\;.
\end{align}
\end{proof}
\end{lem}

\begin{lem}
$\lim_{k\to\infty}F_a\left(\bar{x}^{(N_k)}\right)=\frac{\xi\,h(\xi)}{1-\xi}$
\begin{proof}
We can rewrite
\begin{equation}\label{eq:FNk}
F_a\left(\bar{x}^{(N_k)}\right)=\sum_{n=1}^{N_k} n\,p_{n-1}^{(N_k)}\,w_{n-1}^{(N_k)}\,h\left(w_{n-1}^{(N_k)}\right)\;.
\end{equation}
For any $0\le s\le 1$ we have
\begin{equation}
0\le s \,h(s) = s^a-s \le 2\;,
\end{equation}
and from \eqref{eq:boundqnk} the sums in \eqref{eq:FNk} are dominated by $2\sum_{n=1}^\infty n\left(\frac{1+\xi}{2}\right)^n<\infty$.
Then, from the dominated convergence theorem and Lemma \ref{lem:limq0} we have
\begin{equation}
\lim_{k\to\infty}F_a\left(\bar{x}^{(N_k)}\right)=h(\xi)\sum_{n=1}^{\infty} n\left(1-\xi\right)\xi^n=\frac{\xi\,h(\xi)}{1-\xi}\;.
\end{equation}
\end{proof}
\end{lem}

Since $\mathcal{P}_N\subset\mathcal{P}_{N'}$ for any $N<N'$, the sequence $\left\{\mathcal{F}_{z,a}\left(\bar{x}^{(N)}\right)\right\}_{N\in\mathbb{N}}$ is increasing.
We then have for any $N\in\mathbb{N}$ and any $x\in\mathbb{R}^{N+1}$ with nonnegative components and satisfying \eqref{eq:Trx}
\begin{align}
\mathcal{F}_{z,a}(x) &\le \mathcal{F}_{z,a}\left(\bar{x}^{(N)}\right) \le \sup_{N'\in\mathbb{N}}\mathcal{F}\left(\bar{x}^{(N')}\right) = \lim_{N'\to\infty}\mathcal{F}\left(p^{(N')}\right)\nonumber\\
&= \lim_{k\to\infty}\mathcal{F}\left(p^{(N_k)}\right) = \frac{\xi\,h(\xi)}{1-\xi}+\nu(\xi,a)\,S(\hat{\omega}_\xi) = \mathcal{F}_{z,a}\left(\hat{\omega}_{\xi}\right)\;,
\end{align}
and the claim follows since
\begin{equation}
\hat{\omega}_{\xi}=\frac{\hat{\omega}_z^{1/a}}{\mathrm{Tr}\,\hat{\omega}_z^{1/a}}\;.
\end{equation}

\section{The quantum-limited attenuator has Gaussian maximizers}\label{sec:pq}
\begin{thm}[$p\to q$ norms]\label{thm:pq}
For any $1<p<q$ and any $0 < \lambda < 1$ the $p\to q$ norm of $\mathcal{E}_\lambda$ is achieved by a thermal Gaussian state $\hat{\omega}$ (that depends on $\lambda$, $p$ and $q$), i.e.\ for any operator $\hat{X}$ with $\left\|\hat{X}\right\|_p<\infty$
\begin{equation}\label{eq:pq}
\frac{\left\|\mathcal{E}_\lambda\left(\hat{X}\right)\right\|_q}{\left\|\hat{X}\right\|_p}\le \frac{\left\|\mathcal{E}_\lambda\left(\hat{\omega}\right)\right\|_q}{\left\|\hat{\omega}\right\|_p} = \left\|\mathcal{E}_\lambda\right\|_{p\to q}\;.
\end{equation}
For any $p>1$ and any $0<\lambda<1$ the $p\to p$ norm of $\mathcal{E}_\lambda$ is asymptotically achieved by thermal Gaussian states with infinite temperature, i.e.\ for any operator $\hat{X}$ with $\left\|\hat{X}\right\|_p<\infty$
\begin{equation}\label{eq:ppP}
\frac{\left\|\mathcal{E}_\lambda\left(\hat{X}\right)\right\|_p}{\left\|\hat{X}\right\|_p} \le \lim_{z\to1} \frac{\left\|\mathcal{E}_\lambda\left(\hat{\omega}_z\right)\right\|_p}{\left\|\hat{\omega}_z\right\|_p} = \lambda^\frac{1-p}{p} = \left\|\mathcal{E}_\lambda\right\|_{p\to p}\;.
\end{equation}
For any $1<q<p$ and any $0<\lambda<1$ the $p\to q$ norm of $\mathcal{E}_\lambda$ is infinite and is asymptotically achieved by thermal Gaussian states with infinite temperature:
\begin{equation}
\lim_{z\to1} \frac{\left\|\mathcal{E}_\lambda\left(\hat{\omega}_z\right)\right\|_q}{\left\|\hat{\omega}_z\right\|_p} = \infty = \left\|\mathcal{E}_\lambda\right\|_{p\to q}\;.
\end{equation}
\end{thm}
\begin{rem}
$\|\mathcal{E}_0\|_{p\to q} = \infty$ for any $p,\,q>1$.
Indeed, $\mathcal{E}_0(\hat{\omega}_z) = |0\rangle\langle0|$ for any $0\le z <1$, hence
\begin{equation}
\lim_{z\to1}\frac{\left\|\mathcal{E}_0(\hat{\omega}_z)\right\|_q}{\|\hat{\omega}_z\|_p} = \lim_{z\to1}\frac{1}{{\|\hat{\omega}_z\|_p}} = \infty\;.
\end{equation}
\end{rem}
\begin{rem}
$\mathcal{E}_1 = \mathbb{I}$, hence
\begin{equation}
\|\mathcal{E}_1\|_{p\to q} = \left\{
                               \begin{array}{cc}
                                 1 & \text{if}\,1\le p\le q \\
                                 \infty & \text{if}\,1\le q<p \\
                               \end{array}
                             \right.\;.
\end{equation}
\end{rem}

\subsection{\it 1\textless p\textless q}

\begin{lem}
It is sufficient to prove Theorem \ref{thm:pq} for positive operators.
\begin{proof}
From \cite[Theorem 1]{audenaert2009note}, we can restrict to $\hat{X}$ Hermitian.
From \cite[Eq. (2)]{audenaert2009note}, we have $\left\|\mathcal{E}_\lambda\left(\hat{X}\right)\right\|_q\le\left\|\mathcal{E}_\lambda\left(\left|\hat{X}\right|\right)\right\|_q$, where $\left|\hat{X}\right|$ is the absolute value of $\hat{X}$.
Since $\left\|\hat{X}\right\|_p=\left\|\left|\hat{X}\right|\right\|_p$,
\begin{equation}
\frac{\left\|\mathcal{E}_\lambda\left(\hat{X}\right)\right\|_q}{\left\|\hat{X}\right\|_p}\le \frac{\left\|\mathcal{E}_\lambda\left(\left|\hat{X}\right|\right)\right\|_q}{\left\|\left|\hat{X}\right|\right\|_p}\;,
\end{equation}
and the claim follows.
\end{proof}
\end{lem}

\begin{lem}
It is sufficient to prove Theorem \ref{thm:pq} for positive passive operators.
In other words, if \eqref{eq:pq} holds for any positive passive operator $\hat{X}$ with $\left\|\hat{X}\right\|_p<\infty$, then it holds for any operator $\hat{X}$ with $\left\|\hat{X}\right\|_p<\infty$.
\begin{proof}
From Theorem \ref{thm:maj}, the Fock rearrangement of the input does not change its $p$ norm and increases the $q$ norm of the output, i.e.\ for any $\hat{X}\ge0$ with $\left\|\hat{X}\right\|_p<\infty$
\begin{equation}
\frac{\left\|\mathcal{E}_\lambda\left(\hat{X}\right)\right\|_q}{\left\|\hat{X}\right\|_p}\le \frac{\left\|\mathcal{E}_\lambda\left(\hat{X}^\downarrow\right)\right\|_q}{\left\|\hat{X}^\downarrow\right\|_p}\;.
\end{equation}
\end{proof}
\end{lem}
\begin{lem}
It is sufficient to prove Theorem \ref{thm:pq} for positive passive operators with finite rank.
\begin{proof}
Let us suppose that
\begin{equation}\label{eq:hypC}
\frac{\left\|\mathcal{E}_\lambda\left(\hat{Y}\right)\right\|_q}{\left\|\hat{Y}\right\|_p}\le C
\end{equation}
for any positive passive $\hat{Y}$ with finite rank.
For any $n\in\mathbb{N}$ let
\begin{equation}
\hat{P}_n:=\sum_{k=0}^n|k\rangle\langle k|
\end{equation}
be the projector onto the first $n+1$ Fock states, and let $\hat{X}$ be a generic positive passive operator with $\left\|\hat{X}\right\|_p<\infty$.
Since $\hat{X}$ is diagonal in the Fock basis we have
\begin{equation}
0\le\hat{P}_n\,\hat{X}=\hat{X}\,\hat{P}_n=\hat{P}_n\,\hat{X}\,\hat{P}_n\le \hat{X}
\end{equation}
and
\begin{equation}
0\le\left(\hat{P}_n\,\hat{X}\right)^p=\hat{P}_n\,\hat{X}^p\le \hat{X}^p\;.
\end{equation}
We then have from the dominated convergence theorem
\begin{equation}
\lim_{n\to\infty}\mathrm{Tr}\left[\left(\hat{P}_n\,\hat{X}\right)^p\right]=\lim_{n\to\infty}\mathrm{Tr}\left[\hat{P}_n\,\hat{X}^p\right]=\mathrm{Tr}\,\hat{X}^p\;,
\end{equation}
i.e.
\begin{equation}\label{eq:limn}
\lim_{n\to\infty}\left\|\hat{P}_n\,\hat{X}\right\|_p=\left\|\hat{X}\right\|_p\;.
\end{equation}
Since $\mathcal{E}_\lambda$ is a positive map and preserves the set of Fock-diagonal operators, we also have
\begin{equation}
0\le\mathcal{E}_\lambda\left(\hat{P}_n\,\hat{X}\right)\leq\mathcal{E}_\lambda\left(\hat{X}\right)
\end{equation}
and
\begin{equation}
0\le\mathcal{E}_\lambda\left(\hat{P}_n\,\hat{X}\right)^q\leq\mathcal{E}_\lambda\left(\hat{X}\right)^q\;.
\end{equation}
Fatou's lemma implies
\begin{equation}
\mathrm{Tr}\left[\mathcal{E}_\lambda\left(\hat{X}\right)^q\right] \le \liminf_{n\to\infty}\mathrm{Tr}\left[\mathcal{E}_\lambda\left(\hat{P}_n\,\hat{X}\right)^q\right] \;,
\end{equation}
i.e.
\begin{equation}\label{eq:limE}
\left\|\mathcal{E}_\lambda\left(\hat{X}\right)\right\|_q\le \liminf_{n\to\infty}\left\|\mathcal{E}_\lambda\left(\hat{P}_n\,\hat{X}\right)\right\|_q\;.
\end{equation}
Finally, we have from \eqref{eq:limn} and \eqref{eq:limE}
\begin{equation}
\frac{\left\|\mathcal{E}_\lambda\left(\hat{X}\right)\right\|_q}{\left\|\hat{X}\right\|_p}\le \liminf_{n\to\infty}\frac{\left\|\mathcal{E}_\lambda\left(\hat{P}_n\,\hat{X}\right)\right\|_q}{\left\|\hat{P}_n\,\hat{X}\right\|_p}\le C\;,
\end{equation}
where the last inequality follows from the hypothesis \eqref{eq:hypC} applied to $\hat{Y}=\hat{P}_n\,\hat{X}$.
\end{proof}
\end{lem}

Theorem \ref{thm:pq} follows integrating the logarithmic Sobolev inequality \eqref{logs}.
\begin{lem}\label{lem:S}
For any positive operator $\hat{X}$ with finite rank
\begin{equation}
\frac{d}{da} \ln \norm{ \hat{X}  }_\frac{1}{a} = S_a\left(\hat{X}\right)\;,
\end{equation}
where $S_a$ is as in \eqref{eq:defs}.
\begin{proof}
We have
\begin{align}
\frac{d}{da}\ln\left\|\hat{X}\right\|_\frac{1}{a} &= \frac{d}{da}\left(a\ln\mathrm{Tr}\,\hat{X}^\frac{1}{a}\right) = \ln \mathrm{Tr}\,\hat{X}^\frac{1}{a} - \frac{\mathrm{Tr}\left[\hat{X}^\frac{1}{a}\ln\hat{X}\right]}{a\,\mathrm{Tr}\,\hat{X}^\frac{1}{a}}\nonumber\\
&= -\mathrm{Tr}\left[\frac{\hat{X}^\frac{1}{a}}{\mathrm{Tr}\,\hat{X}^\frac{1}{a}}\ln\frac{\hat{X}^\frac{1}{a}}{\mathrm{Tr}\,\hat{X}^\frac{1}{a}}\right] = S\left(\frac{\hat{X}^\frac{1}{a}}{\mathrm{Tr}\,\hat{X}^\frac{1}{a}}\right)\;.
\end{align}
\end{proof}
\end{lem}
\begin{lem}\label{lem:pq}
Let us fix $T>0$ and $0<z_0<1$.
We define for any $0\le t\le T$
\begin{equation}\label{eq:defzt}
z(t) = \frac{\mathrm{e}^{-t}z}{1-\left(1-\mathrm{e}^{-t}\right)z}\;,
\end{equation}
such that
\begin{equation}\label{eq:defwt}
\hat{\omega}_{z(t)}=\mathrm{e}^{t\mathcal{L}}\left(\hat{\omega}_{z_0}\right)\;,
\end{equation}
and let $a(t)$ satisfy
\begin{equation}\label{eq:ode-a}
0<a(t)<1\;,\qquad \frac{d}{dt}a(t)=\mu(z(t),a(t))\;,
\end{equation}
with $\mu(z,a)$ as in \eqref{eq:mu00}.
Then for any positive passive operator $\hat{X}$ with finite rank
\begin{equation}\label{eq:inequality-lemma-12}
\frac{\left\|\mathrm{e}^{T\mathcal{L}}\left(\hat{X}\right)\right\|_\frac{1}{a(T)}}{\left\|\hat{X}\right\|_\frac{1}{a(0)}}\le \frac{\left\|\mathrm{e}^{T\mathcal{L}}\left(\hat{\omega}_{z_0}\right)\right\|_\frac{1}{a(T)}}{\left\|\hat{\omega}_{z_0}\right\|_\frac{1}{a(0)}}\;.
\end{equation}
\begin{proof}
We have with the help of Lemma \ref{lem:S} and of the differential equation \eqref{eq:ode-a} for $a(t)$
\begin{align}\label{eq:dt}
\frac{d}{dt}\ln\left\|\mathrm{e}^{t\mathcal{L}}\left(\hat{X}\right)\right\|_\frac{1}{a(t)} &=F_{a(t)}\left(\mathrm{e}^{t\mathcal{L}}\left(\hat{X}\right)\right) + S_{a(t)}\left(\mathrm{e}^{t\mathcal{L}}\left(\hat{X}\right)\right)\frac{d}{dt}a(t)\nonumber\\
&= \mathcal{F}_{z(t),a(t)}\left(\mathrm{e}^{t\mathcal{L}}\left(\hat{X}\right)\right) \le \mathcal{F}_{z(t),a(t)}\left(\hat{\omega}_{z(t)}\right)\nonumber\\
&=F_{a(t)}\left(\mathrm{e}^{t\mathcal{L}}\left(\hat{\omega}_{z_0}\right)\right) +S_{a(t)}\left(\mathrm{e}^{t\mathcal{L}}\left(\hat{\omega}_{z_0}\right)\right)\frac{d}{dt}a(t)\nonumber\\
&=\frac{d}{dt}\ln\left\|\mathrm{e}^{t\mathcal{L}}\left(\hat{\omega}_{z_0}\right)\right\|_\frac{1}{a(t)}\;,
\end{align}
where we have used the logarithmic Sobolev inequality \eqref{logs} with $\hat{X}$ replaced by $\mathrm{e}^{t\mathcal{L}}\left(\hat{X}\right)$.
Inequality \eqref{eq:inequality-lemma-12} then follows integrating \eqref{eq:dt} between $t=0$ and $t=T$.
\end{proof}
\end{lem}

The claim \eqref{eq:pq} then follows from
\begin{prop}
For any $1<p<q$ and any $T >0$ there is $0<z_0<1$ such that the solution to \eqref{eq:ode-a}, with $a(0) = 1/p$ and $z(t)$ given by \eqref{eq:defwt}, satisfies $a(T) = 1/q$.
\begin{proof}
We can extend the function $(z,a)\mapsto\mu(z,a)$ to $z=1$ with $\mu(1,a):=0$.
This extension is continuously differentiable for $0<z\le1$ and $0<a<1$.
We then have $a(T)=a(0)=1/p$ for $z_0=1$.
Since $0<1/q<1/p$, by continuous dependence of solutions of differential equations with parameters, e.g., \cite[Theorem 3.1]{hartman1982ordinary}, it is sufficient to show that there exists some $0<z_0<1$ such that $a(T)=0$.

We have for $0<a<1$ and $0<z<1$
\begin{equation}\label{eq:muphi}
\mu(z,a)\le \frac{a\,z^\frac{a-1}{a}-1}{\ln z^\frac{1}{a}}:=\phi(z,a)\;.
\end{equation}
The function $a\mapsto\phi(z,a)$ is increasing for $z\le \mathrm{e}^{-2}$.
Indeed,
\begin{equation}
\frac{\partial}{\partial a}\phi(z,a)= z^\frac{a-1}{a}\frac{2a-z^\frac{1-a}{a}+\ln z}{\ln z}\ge z^\frac{a-1}{a}\frac{2+\ln z}{\ln z}\ge0\;.
\end{equation}
Moreover, the function $z\mapsto\phi(z,a)$ is increasing for $z\le \mathrm{e}^\frac{a}{a-1}$.
Indeed,
\begin{equation}
\frac{\partial}{\partial z}\phi(z,a)= a\frac{z^\frac{1-a}{a}-a-\left(1-a\right)\ln z}{z^\frac{1}{a}\ln^2z}\ge a\frac{-a-\left(1-a\right)\ln z}{z^\frac{1}{a}\ln^2z}\ge0\;.
\end{equation}
From Lemma \ref{lem:mu} $\mu(z,a)<0$, hence the function $t\mapsto a(t)$ is decreasing and therefore
\begin{equation}\label{eq:a}
a(t)\le a(0)=1/p\;.
\end{equation}
From \eqref{eq:defzt}, the function $t\mapsto z(t)$ is decreasing, and therefore
\begin{equation}\label{eq:zd}
z(t)\le z_0\;.
\end{equation}
Then, choosing $z_0\le \mathrm{e}^{-2}$ and $z_0\le \mathrm{e}^\frac{1}{1-p}$ we have $z(t)\le \mathrm{e}^{-2}$ and $z(t)\le \mathrm{e}^\frac{1}{1-p}$ for any $t>0$.
Then,
\begin{equation}
\mu(z(t),a(t))\le\phi(z(t),a(t))\le \phi(z(t),1/p)\le\phi(z_0,1/p)\;,
\end{equation}
where the first inequality follows from \eqref{eq:muphi}, the second inequality follows from \eqref{eq:a} since the function $a\mapsto\phi(z,a)$ is increasing for $z\le \mathrm{e}^{-2}$, and the last inequality follows from \eqref{eq:zd} since the function $z\mapsto\phi(z,a)$ is increasing for $z\le \mathrm{e}^\frac{a}{a-1}$.
We then have
\begin{equation}
a(T)\le \frac{1}{p}+\phi(z_0,1/p)\;T\;.
\end{equation}
Since $\lim_{z\to0}\phi(z,1/p)=-\infty$ for any $p>1$, for sufficiently small $z_0$ we will surely have $a(T)=0$.
\end{proof}
\end{prop}

\subsection{\it 1\textless p=q}
With the same proof used in the case $1<p<q$, we can restrict to $\hat{X}$ positive passive with finite rank.
Integrating \eqref{eq:Fa} with $a=1/p$ between $t=0$ and $t=-\ln\lambda$ we get
\begin{equation}
\frac{\left\|\mathcal{E}_\lambda\left(\hat{X}\right)\right\|_p}{\left\|\hat{X}\right\|_p} \le \lambda^\frac{1-p}{p}\;.
\end{equation}
We can compute for any $0<z<1$
\begin{equation}
\left\|\hat{\omega}_z\right\|_p = \frac{1-z}{\left(1-z^p\right)^\frac{1}{p}} := f_p(z)\;.
\end{equation}
We have for $\epsilon\to0$
\begin{equation}\label{eq:deffp}
f_p(1-\epsilon) = \epsilon^\frac{p-1}{p}\frac{1+\mathcal{O}(\epsilon)}{p^\frac{1}{p}}\;,
\end{equation}
and using \eqref{eq:zlambda}
\begin{equation}
\lim_{\epsilon\to0}\frac{\left\|\mathcal{E}_\lambda\left(\hat{\omega}_{1-\epsilon}\right)\right\|_p}{\left\|\hat{\omega}_{1-\epsilon}\right\|_p} = \lim_{\epsilon\to0}\frac{f_p\left(1-\frac{\epsilon}{\lambda+\left(1-\lambda\right)\epsilon}\right)}{f_p(1-\epsilon)} = \lambda^\frac{1-p}{p}\;.
\end{equation}

\subsection{\it 1\textless q\textless p}
We have using \eqref{eq:zlambda} and \eqref{eq:deffp}
\begin{equation}
\lim_{\epsilon\to0}\frac{\left\|\mathcal{E}_\lambda\left(\hat{\omega}_{1-\epsilon}\right)\right\|_q}{\left\|\hat{\omega}_{1-\epsilon}\right\|_p} = \lim_{\epsilon\to0}\frac{f_q\left(1-\frac{\epsilon}{\lambda+\left(1-\lambda\right)\epsilon}\right)}{f_p(1-\epsilon)} = \infty\;.
\end{equation}

\section{The quantum-limited amplifier has Gaussian maximizers}\label{sec:ampl}
\begin{thm}\label{thm:ampl}
For any $1<p<q$ and any $\kappa>1$ the $p\to q$ norm of $\mathcal{A}_\kappa$ is achieved by a thermal Gaussian state $\hat{\omega}$ (that depends on $\kappa$, $p$ and $q$), i.e.\ for any operator $\hat{X}$ with $\left\|\hat{X}\right\|_p<\infty$
\begin{equation}
\frac{\left\|\mathcal{A}_\kappa\left(\hat{X}\right)\right\|_q}{\left\|\hat{X}\right\|_p}\le \frac{\left\|\mathcal{A}_\kappa\left(\hat{\omega}\right)\right\|_q}{\left\|\hat{\omega}\right\|_p} = \left\|\mathcal{A}_\kappa\right\|_{p\to q}\;.
\end{equation}
For any $p>1$ and any $\kappa>1$ the $p\to p$ norm of $\mathcal{A}_\kappa$ is asymptotically achieved by thermal Gaussian states with infinite temperature, i.e.\ for any operator $\hat{X}$ with $\left\|\hat{X}\right\|_p<\infty$
\begin{equation}
\frac{\left\|\mathcal{A}_\kappa\left(\hat{X}\right)\right\|_p}{\left\|\hat{X}\right\|_p} \le \lim_{z\to1} \frac{\left\|\mathcal{A}_\kappa\left(\hat{\omega}_z\right)\right\|_p}{\left\|\hat{\omega}_z\right\|_p} = \kappa^\frac{1-p}{p} = \left\|\mathcal{A}_\kappa\right\|_{p\to p}\;.
\end{equation}
For any $1<q<p$ and any $\kappa>1$ the $p\to q$ norm of $\mathcal{A}_\kappa$ is infinite and is asymptotically achieved by thermal Gaussian states with infinite temperature:
\begin{equation}
\lim_{z\to1} \frac{\left\|\mathcal{A}_\kappa\left(\hat{\omega}_z\right)\right\|_q}{\left\|\hat{\omega}_z\right\|_p} = \infty = \left\|\mathcal{A}_\kappa\right\|_{p\to q}\;.
\end{equation}
\end{thm}

With the same argument we used for the attenuator, we can reduce to $\hat{X}$ positive, passive and with finite rank.
In the proof we will repeatedly use
\begin{thm}[H\"older's inequality]
For any two commuting positive matrices $\hat{A}$ and $\hat{B}$ and any $p>1$
\begin{equation}
\mathrm{Tr}\left[\hat{A}\;\hat{B}\right]\le\left\|\hat{A}\right\|_p\left\|\hat{B}\right\|_{p'}\;,
\end{equation}
where $p':=\frac{p}{p-1}$ is the conjugate exponent of $p$.
\end{thm}

\subsection{1\textless p\textless q}
We recalling that the dual of $\mathcal{A}_\kappa$ is $\frac{1}{\kappa}\mathcal{E}_\frac{1}{\kappa}$ (Proposition \ref{prop:dual}), and that both $\mathcal{A}_\kappa$ and $\mathcal{E}_\frac{1}{\kappa}$ preserve the set of Fock-diagonal states.
Then, H\"older's inequality gives
\begin{align}
\left\|\mathcal{A}_\kappa\left(\hat{X}\right)\right\|_q^q &= \frac{1}{\kappa}\mathrm{Tr}\left[\hat{X}\;\mathcal{E}_\frac{1}{\kappa}\left(\mathcal{A}_\kappa\left(\hat{X}\right)^{q-1}\right)\right]\nonumber\\
&\le \frac{1}{\kappa}\left\|\hat{X}\right\|_p\left\|\mathcal{E}_\frac{1}{\kappa}\left(\mathcal{A}_\kappa\left(\hat{X}\right)^{q-1}\right)\right\|_{p'}\nonumber\\
&\le \left\|\hat{X}\right\|_p\left\|\mathcal{A}_\kappa\left(\hat{X}\right)^{q-1}\right\|_{q'} \frac{\left\|\mathcal{E}_\frac{1}{\kappa}\left(\hat{\omega}\right)\right\|_{p'}}{\kappa\left\|\hat{\omega}\right\|_{q'}}\nonumber\\
&=\left\|\hat{X}\right\|_p\left\|\mathcal{A}_\kappa\left(\hat{X}\right)\right\|_q^{q-1} \frac{\left\|\mathcal{E}_\frac{1}{\kappa}\left(\hat{\omega}\right)\right\|_{p'}}{\kappa\left\|\hat{\omega}\right\|_{q'}}\;,
\end{align}
where the first identity follows from \eqref{eq:dual} and $\hat{\omega}$ is the Gaussian state that achieves the $q'\to p'$ norm of $\mathcal{E}_{1/\kappa}$ according to Theorem \ref{thm:pq} ($1<p<q$ implies $1<q'<p'$).
Then,
\begin{equation}\label{eq:Apq}
\frac{\left\|\mathcal{A}_\kappa\left(\hat{X}\right)\right\|_q}{\left\|\hat{X}\right\|_p}\le \frac{\left\|\mathcal{E}_\frac{1}{\kappa}\left(\hat{\omega}\right)\right\|_{p'}}{\kappa \left\|\hat{\omega}\right\|_{q'}}\;.
\end{equation}

Since the dual of $\frac{1}{\kappa}\mathcal{E}_\kappa$ is $\mathcal{A}_\kappa$ again, H\"older's inequality gives
\begin{align}
\frac{1}{\kappa}\left\|\mathcal{E}_\frac{1}{\kappa}\left(\hat{\omega}\right)\right\|_{p'}^{p'} &= \mathrm{Tr}\left[\hat{\omega}\;\mathcal{A}_\kappa\left(\mathcal{E}_\frac{1}{\kappa}\left(\hat{\omega}\right)^{p'-1}\right)\right]\nonumber\\
&\le \left\|\hat{\omega}\right\|_{q'}\left\|\mathcal{A}_\kappa\left(\mathcal{E}_\frac{1}{\kappa}\left(\hat{\omega}\right)^{p'-1}\right)\right\|_q\;,
\end{align}
hence
\begin{align}\label{eq:Awpq}
\frac{\left\|\mathcal{E}_\frac{1}{\kappa}\left(\hat{\omega}\right)\right\|_{p'}}{\kappa\left\|\hat{\omega}\right\|_{q'}} &\le \frac{\left\|\mathcal{A}_\kappa\left(\mathcal{E}_\frac{1}{\kappa}\left(\hat{\omega}\right)^{p'-1}\right)\right\|_q} {\left\|\mathcal{E}_\frac{1}{\kappa}\left(\hat{\omega}\right)\right\|_{p'}^{p'-1}}\nonumber\\
&= \frac{\left\|\mathcal{A}_\kappa\left(\mathcal{E}_\frac{1}{\kappa}\left(\hat{\omega}\right)^{p'-1}\right)\right\|_q} {\left\|\mathcal{E}_\frac{1}{\kappa}\left(\hat{\omega}\right)^{p'-1}\right\|_p}\;.
\end{align}
Putting together \eqref{eq:Apq} and \eqref{eq:Awpq} we get that the $p\to q$ norm of $\mathcal{A}_\kappa$ is achieved by the quantum state
\begin{equation}
\hat{\omega}':=\frac{\mathcal{E}_\frac{1}{\kappa}\left(\hat{\omega}\right)^{p'-1}}{\mathrm{Tr}\;\mathcal{E}_\frac{1}{\kappa}\left(\hat{\omega}\right)^{p'-1}}\;.
\end{equation}
Since $\hat{\omega}$ is a thermal Gaussian state, also $\mathcal{E}_\frac{1}{\kappa}\left(\hat{\omega}\right)$ is a (different) thermal Gaussian state.
Finally, from \eqref{eq:omegaz} any positive power of a thermal Gaussian state is proportional to another thermal Gaussian state, hence also $\hat{\omega}'$ is a thermal Gaussian state.

\subsection{1\textless p=q}
Recalling that the dual of $\mathcal{A}_\kappa$ is $\frac{1}{\kappa}\mathcal{E}_\frac{1}{\kappa}$, and that both $\mathcal{A}_\kappa$ and $\mathcal{E}_\frac{1}{\kappa}$ preserve the set of Fock-diagonal states, H\"older's inequality gives
\begin{align}
\left\|\mathcal{A}_\kappa\left(\hat{X}\right)\right\|_p^p &= \frac{1}{\kappa}\mathrm{Tr}\left[\hat{X}\;\mathcal{E}_\frac{1}{\kappa}\left(\mathcal{A}_\kappa\left(\hat{X}\right)^{p-1}\right)\right]\nonumber\\
&\le \frac{1}{\kappa}\left\|\hat{X}\right\|_p\left\|\mathcal{E}_\frac{1}{\kappa}\left(\mathcal{A}_\kappa\left(\hat{X}\right)^{p-1}\right)\right\|_{p'}\nonumber\\
&\le \left\|\hat{X}\right\|_p\left\|\mathcal{A}_\kappa\left(\hat{X}\right)^{p-1}\right\|_{p'} \frac{\left\|\mathcal{E}_\frac{1}{\kappa}\right\|_{p'\to p'}}{\kappa}\nonumber\\
&=\left\|\hat{X}\right\|_p\left\|\mathcal{A}_\kappa\left(\hat{X}\right)\right\|_p^{p-1} \kappa^\frac{1-p}{p}\;,
\end{align}
where we have used \eqref{eq:ppP}, hence
\begin{equation}
\left\|\mathcal{A}_\kappa\right\|_{p\to p}\le \kappa^\frac{1-p}{p}\;.
\end{equation}
Moreover, using \eqref{eq:z'} and \eqref{eq:deffp} we have
\begin{equation}
\lim_{\epsilon\to0}\frac{\left\|\mathcal{A}_\kappa\left(\hat{\omega}_{1-\epsilon}\right)\right\|_p}{\left\|\hat{\omega}_{1-\epsilon}\right\|_p} = \lim_{\epsilon\to0}\frac{f_p\left(1-\frac{\epsilon}{\kappa}\right)}{f_p(1-\epsilon)} = \kappa^\frac{1-p}{p}\;.
\end{equation}

\subsection{1\textless q\textless p}
We have from \eqref{eq:z'} and \eqref{eq:deffp}
\begin{equation}
\lim_{\epsilon\to0}\frac{\left\|\mathcal{A}_\kappa\left(\hat{\omega}_{1-\epsilon}\right)\right\|_q}{\left\|\hat{\omega}_{1-\epsilon}\right\|_p} = \lim_{\epsilon\to0}\frac{f_q\left(1-\frac{\epsilon}{\kappa}\right)}{f_p(1-\epsilon)} = \infty\;.
\end{equation}

\section{Thermal quantum Gaussian channels}\label{sec:thermal}
The one-mode phase-covariant quantum Gaussian channel are the quantum channels resulting from the composition of a quantum-limited amplifier with a quantum-limited attenuator \cite{mari2014quantum,giovannetti2015solution,holevo2015gaussian,garcia2012majorization}.
The following upper bound to their $p\to q$ norms holds:
\begin{prop}\label{prop:bound}
For any $0\le\lambda\le 1$, any $\kappa\ge1$ and any $p,\,q\ge1$,
\begin{equation}\label{eq:boundc}
\left\|\mathcal{A}_\kappa\circ\mathcal{E}_\lambda\right\|_{p\to q} \le \inf_{r\ge 1}\left(\left\|\mathcal{A}_\kappa\right\|_{r\to q}\left\|\mathcal{E}_\lambda\right\|_{p\to r}\right)\;.
\end{equation}
\end{prop}
\begin{proof}
For any operator $\hat{X}$ with $\left\|\hat{X}\right\|_p<\infty$ and for any $r\ge1$,
\begin{equation}
\left\|\left(\mathcal{A}_\kappa\circ\mathcal{E}_\lambda\right)\left(\hat{X}\right)\right\|_q \le \left\|\mathcal{A}_\kappa\right\|_{r\to q}\left\|\mathcal{E}_\lambda\left(\hat{X}\right)\right\|_r \le \left\|\mathcal{A}_\kappa\right\|_{r\to q}\left\|\mathcal{E}_\lambda\right\|_{p\to r}\left\|\hat{X}\right\|_p\;.
\end{equation}
\end{proof}
The lack of an analytic expression for $\left\|\mathcal{A}_\kappa\right\|_{r\to q}$ and $\left\|\mathcal{E}_\lambda\right\|_{p\to r}$ makes the computation of the right-hand side of \eqref{eq:boundc} difficult.
However, we conjecture that the right-hand side of \eqref{eq:boundc} coincides with the $p\to q$ norm of $\mathcal{A}_\kappa\circ\mathcal{E}_\lambda$.

\section{The constrained minimum output entropy conjecture}\label{secMOE}
A crucial consequence of Theorem \ref{thm:ampl} is the constrained minimum output entropy conjecture for the one-mode phase-covariant quantum Gaussian channels.
\begin{thm}[constrained minimum output entropy conjecture {\cite[Theorem 4]{de2016gaussiannew}}]\label{thm:MOE}
Quantum thermal Gaussian input states minimize the output von Neumann entropy of any one-mode phase-covariant quantum Gaussian channel among all the input states with a given von Neumann entropy.
In other words, for any one-mode phase-covariant quantum Gaussian channel $\Phi$, any $0\le z<1$ and any one-mode quantum state $\hat{\rho}$ with $S(\hat{\rho}) = S(\hat{\omega}_z)$,
\begin{equation}
S(\Phi(\hat{\rho}))\ge S(\Phi(\hat{\omega}_z))\;.
\end{equation}
\end{thm}
We provide here a sketch of the proof of Theorem \ref{thm:MOE}.
The reader can find the complete proof in Ref. \cite{de2016gaussiannew}.

Theorem \ref{thm:MOE} has been proven in Ref. \cite{de2016gaussian} in the case where $\Phi$ is a one-mode quantum-limited attenuator.
The proof of Ref. \cite{de2016gaussian} is based on an isoperimetric inequality and does not use the $p\to q$ norms.
Since any one-mode phase-covariant quantum Gaussian channel can be decomposed as a quantum-limited attenuator followed by a quantum-limited amplifier \cite{mari2014quantum,giovannetti2015solution,holevo2015gaussian,garcia2012majorization}, it is sufficient to prove Theorem \ref{thm:MOE} when $\Phi$ is a one-mode quantum-limited amplifier (see \cite{de2016gaussiannew} for the details).
For this purpose, it is convenient to rephrase Theorem \ref{thm:ampl} in the following way.
\begin{lem}[{\cite[Lemma 1]{de2016gaussiannew}}]
Let us fix $\kappa\ge1$ and $0 < z <1$.
Then, for any $1<q<\frac{3}{2}$ there exists $1 < p<q$ such that the $p\to q$ norm of $\mathcal{A}_{\kappa}$ is achieved by $\hat{\omega}_z$, and for any one-mode quantum state $\hat{\rho}$
\begin{equation}\label{eq:pqA}
\frac{\|\mathcal{A}_{\kappa}(\hat{\rho})\|_q}{\|\hat{\rho}\|_p} \le \|\mathcal{A}_{\kappa}\|_{p\to q} = \frac{\|\mathcal{A}_{\kappa}(\hat{\omega}_z)\|_q}{\|\hat{\omega}_z\|_p}\;.
\end{equation}
\end{lem}
For any $p>1$, the R\'enyi $p$-entropy \cite{holevo2013quantum} of the quantum state $\hat{\rho}$ is
\begin{equation}
S_p(\hat{\rho}) = \frac{p}{1-p}\ln\|\hat{\rho}\|_p
\end{equation}
and satisfies
\begin{equation}\label{eq:limR}
\lim_{p\to 1} S_p(\hat{\rho}) = S(\hat{\rho})\;.
\end{equation}
Let now $\hat{\rho}$ be a one-mode quantum state with
\begin{equation}
S(\hat{\rho}) = S(\hat{\omega}_z)\;.
\end{equation}
Rewriting \eqref{eq:pqA} in terms of the R\'enyi entropies we get
\begin{equation}
S_q(\mathcal{A}_{\kappa}(\hat{\rho})) \ge S_q(\mathcal{A}_{\kappa}(\hat{\omega}_z)) + \frac{p-1}{q-1}\frac{q}{p}\left(S_p(\hat{\rho}) - S_p(\hat{\omega}_z)\right)\;.
\end{equation}
Taking the limit $q\to 1$ and recalling \eqref{eq:limR} we get the claim
\begin{equation}
S(\mathcal{A}_{\kappa}(\hat{\rho})) \ge S(\mathcal{A}_{\kappa}(\hat{\omega}_z))\;.
\end{equation}

\section{The thinning}\label{secthinning}
The thinning \cite{renyi1956characterization} is the map acting on classical probability distributions on the set of natural numbers that is the discrete analogue of the continuous rescaling operation on positive real numbers.
\begin{defn}[Thinning]
Let $N$ be a random variable with values in $\mathbb{N}$.
The thinning with parameter $0\leq\lambda\leq1$ is defined as
\begin{equation}
T_\lambda(N)=\sum_{i=1}^N B_i\;,
\end{equation}
where the $\{B_n\}_{n\in\mathbb{N}^+}$ are independent Bernoulli variables with parameter $\lambda$, i.e.\ each $B_i$ is one with probability $\lambda$, and zero with probability $1-\lambda$.
\end{defn}
From a physical point of view, the thinning can be understood as follows: each incoming photon has probability $\lambda$ of being transmitted, and $1-\lambda$ of being reflected or absorbed.
Let $N$ be the random variable associated to the number of incoming photons, and $\{p_n\}_{n\in\mathbb{N}}$ its probability distribution, i.e.\ $p_n$ is the probability that $N=n$ (i.e.\ that $n$ photons are sent).
Then, $T_\lambda(p)$ is the probability distribution of the number of transmitted photons.
It is easy to show that
\begin{equation}\label{Tn}
\left[T_\lambda(p)\right]_n=\sum_{k=0}^\infty r_{n|k}\;p_k\;,
\end{equation}
where the transition probabilities $r_{n|k}$ are given by
\begin{equation}\label{rnk}
r_{n|k}=\binom{k}{n}\lambda^n(1-\lambda)^{k-n}\;,
\end{equation}
and vanish for $k<n$.

The thinning coincides with the restriction of the attenuator to input states diagonal in the Fock basis:
\begin{thm}\label{thinatt}
Let $\mathcal{E}_\lambda$ and $T_\lambda$ be the quantum-limited attenuator and the thinning of parameter $0\leq\lambda\leq1$, respectively.
Then for any probability distribution $p$ on $\mathbb{N}$
\begin{equation}
\mathcal{E}_\lambda\left(\sum_{n=0}^\infty p_n\;|n\rangle\langle n|\right)=\sum_{n=0}^\infty \left[T_\lambda(p)\right]_n\;|n\rangle\langle n|\;.
\end{equation}
\begin{proof}
See \cite[Theorem 56]{de2015passive}.
\end{proof}
\end{thm}
Thanks to Theorem \ref{thinatt}, our main result applies also to the thinning with the usual $l^p$ norms:
\begin{defn}[$l^p$ norm]
For any $p\ge1$, the $l^p$ norm of a sequence of complex numbers $\left\{x_n\right\}_{n\in\mathbb{N}}$ is
\begin{equation}
\left\|x\right\|_p = \left(\sum_{n\in\mathbb{N}}|x_n|^p\right)^\frac{1}{p}\;.
\end{equation}
\end{defn}
\begin{thm}[$l^p\to l^q$ norms]\label{thm:pqT}
For any $1<p<q$ and any $0<\lambda<1$ the $l^p\to l^q$ norm of $T_\lambda$ is achieved by a geometric probability distribution
\begin{equation}
\omega_n(z)=\left(1-z\right)z^n\;,\qquad n\in\mathbb{N}
\end{equation}
with $0<z<1$ depending on $\lambda$, $p$ and $q$, i.e.\ for any sequence $\{x_n\}_{n\in\mathbb{N}}$ of complex numbers with $\left\|x\right\|_p<\infty$
\begin{equation}
\frac{\left\|T_\lambda\left(x\right)\right\|_q}{\left\|x\right\|_p}\le \frac{\left\|T_\lambda\left(\omega\right)\right\|_q}{\left\|\omega\right\|_p} = \left\|T_\lambda\right\|_{p\to q}\;.
\end{equation}
For any $p>1$ and any $0<\lambda<1$ the $l^p\to l^p$ norm of $T_\lambda$ is asymptotically achieved by geometric probability distributions with $z\to1$, i.e.\ for any sequence $\{x_n\}_{n\in\mathbb{N}}$ of complex numbers with $\left\|x\right\|_p<\infty$
\begin{equation}
\frac{\left\|T_\lambda\left(x\right)\right\|_p}{\left\|x\right\|_p} \le \lim_{z\to1} \frac{\left\|T_\lambda\left(\omega(z)\right)\right\|_p}{\left\|\omega(z)\right\|_p} = \lambda^\frac{1-p}{p} = \left\|T_\lambda\right\|_{p\to p}\;.
\end{equation}
For any $1<q<p$ and any $0<\lambda<1$ the $l^p\to l^q$ norm of $T_\lambda$ is infinite and is asymptotically achieved by geometric probability distributions with $z\to1$:
\begin{equation}
\lim_{z\to1} \frac{\left\|T_\lambda\left(\omega(z)\right)\right\|_q}{\left\|\omega(z)\right\|_p} = \infty = \left\|T_\lambda\right\|_{p\to q}\;.
\end{equation}
\end{thm}

\section{Conclusions}\label{sec:concl}
We have proven that Gaussian states achieve the $p\to q$ norms of the two building blocks of one-mode phase-covariant quantum Gaussian channels: the quantum-limited attenuator and amplifier.
This fundamental result proves a longstanding conjecture, which was open since 2006 \cite{holevo2006multiplicativity}.

Our result has led to the proof \cite{de2016gaussiannew} of the constrained minimum output entropy conjecture, which was open since 2008 and states that Gaussian states minimize the output entropy of any one-mode phase-covariant quantum Gaussian channel among all the input states with a given entropy.
The logarithm of the $p$-norm of a quantum state is proportional to its R\'enyi $p$-entropy \cite{holevo2013quantum}.
The R\'enyi entropies play a key role in quantum information processing with finite resources \cite{tomamichel2015quantum}.
Our result might then also find application in this scenario.

Our proof technique cannot be directly applied to the one-mode phase-covariant quantum Gaussian channels that are not quantum-limited.
Indeed, the application of the Karush--Kuhn--Tucker conditions heavily relies on the restriction to input states diagonal in the Fock basis and with finite rank and on the property that the quantum-limited attenuator preserves this class of states.
Indeed, it has been shown in the case of the minimization of the output entropy that for any other phase-covariant quantum Gaussian channel, quantum Gaussian states are not the only solution to the Karush--Kuhn--Tucker conditions \cite{qi2017minimum}.
The quantum-limited amplifier is an exception since its $p\to q$ norms are determined by the norms of the quantum-limited attenuator via a duality argument.
We have also been able to prove an upper bound to the $p\to q$ norms of any one-mode phase-covariant quantum Gaussian channel, and we conjecture that this upper bound is optimal.

The fundamental property stating that the output generated by the vacuum input state majorizes the output generated by any other input state holds for any multimode phase-covariant quantum Gaussian channel \cite{giovannetti2015majorization,holevo2015gaussian}.
We then conjecture that multimode Gaussian states achieve the $p\to q$ norms of multimode phase-covariant quantum Gaussian channels.
This result would imply the multiplicativity of such norms \cite{holevo2015gaussian}
\begin{equation}
\left\|\Phi\otimes\Phi'\right\|_{p\to q}=\left\|\Phi\right\|_{p\to q}\left\|\Phi'\right\|_{p\to q}
\end{equation}
for any two multimode phase-covariant quantum Gaussian channels $\Phi$ and $\Phi'$, proven so far only for $p=1$ \cite{holevo2015gaussian} and $p=q$ \cite{frank2017norms}.
While the multiplicativity of the norms of any two classical integral kernels can be proven via the Minkowski inequality, in the quantum case entanglement makes it a highly nontrivial property that does not necessarily hold (see \cite{holevo2006multiplicativity} and references therein for a review).

Our proof relies on the majorization results of \cite{de2015passive}, that fail in the multimode scenario (see \cite[section IV.A]{de2016passive}).
The multimode extension will be an open challenge for the future, that could exploit the recently discovered relation between quantum Gaussian semigroups and optimal mass transport \cite{carlen2014analog,carlen2016gradient}.
We conjecture that quantum Gaussian states are the only operators that achieve the $p\to q$ norms of quantum Gaussian channels, as it is for classical Gaussian integral kernels \cite{lieb1990Gaussian}.
The proof of this conjecture will be another open challenge for the future.

\subsection*{Acknowledgment}
GdP acknowledges financial support from the European Research Council (ERC Grant Agreements Nos. 337603 and 321029), the Danish Council for Independent Research (Sapere Aude), VILLUM FONDEN via the QMATH Centre of Excellence (Grant No. 10059), and the Marie Sk\l odowska-Curie Action GENIUS (Grant No. 792557).

\includegraphics[width=0.05\textwidth]{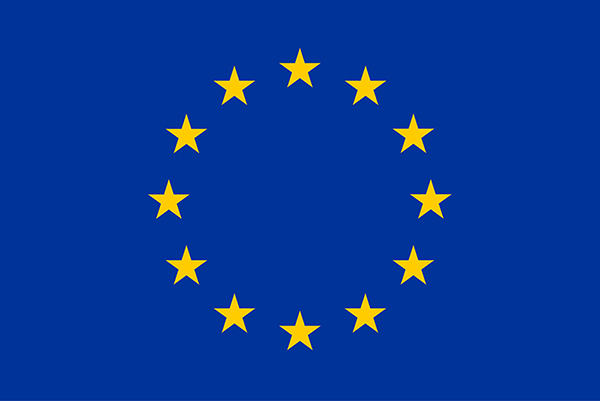}
This project has received funding from the European Union's Horizon 2020 research and innovation programme under the Marie Sk\l odowska-Curie grant agreement No. 792557.

\appendix
\section{Auxiliary Lemmas}
\begin{lem}\label{lem:mu}
The function
\begin{equation}\label{eq:mu3}
\mu(z,a):=\frac{a\,z^\frac{a-1}{a}-1+\left(1-a\right)z}{\ln z^\frac{1}{a}}
\end{equation}
is negative for any $0<a<1$ and any $0<z<1$.
\end{lem}
\begin{proof}
Since $0<a<1$, the function $z\mapsto a z^\frac{a-1}{a}$ is convex and
\begin{equation}
a\,z^\frac{a-1}{a} > a + \left(a-1\right)\left(z-1\right) = 1 - \left(1-a\right)z\;.
\end{equation}
Then, the numerator in \eqref{eq:mu3} is positive.
Since $0<z<1$, the denominator in \eqref{eq:mu3} is negative, and the claim follows.
\end{proof}

\begin{lem}\label{lem:diff}
For any $0<a<1$, the function
\begin{equation}
(x,y)\mapsto x\left|y\right|^{\frac{1}{a}-1}
\end{equation}
is differentiable in $(0,0)$ with a vanishing differential.
\end{lem}
\begin{proof}
\begin{equation}
\limsup_{(x,y)\to(0,0)}\frac{\left|x\right|\left|y\right|^{\frac{1}{a}-1}}{\sqrt{x^2+y^2}} \le \limsup_{(x,y)\to(0,0)}\left|y\right|^{\frac{1}{a}-1} = 0\;.
\end{equation}
\end{proof}
\section{Karush--Kuhn--Tucker conditions}
\begin{thm}[Karush--Kuhn--Tucker conditions {\cite[Theorem 7.2.9]{borwein2013convex}}]\label{thm:KKT}
Let $f$, $\{\psi_n\}_{n=0}^N$, $\phi:\mathbb{R}^M\to\mathbb{R}$ be continuous functions {\bf(a)}.
Let $\bar{x}\in\mathbb{R}^M$ be a local maximizer for $f(x)$ on the domain $\Omega\subset \mathbb{R}^M$ defined by the constraints
\begin{equation}
\psi_n(x)\ge0\quad \forall\,n=0,\,\ldots,\,N,\qquad \phi(x)=0\;.
\end{equation}
Let $f$, $\{\psi_n\}_{n=0}^N$ and $\phi$ be differentiable in $\bar{x}$ {\bf(b)} with $\nabla \phi(\bar{x})\neq 0$ {\bf(c)}.
Let
\begin{equation}
I=\left\{n\in\{0,\,\ldots,\,N\}:\psi_n(\bar{x})=0\right\}\;,
\end{equation}
and let us suppose that the gradients $\nabla \phi(\bar{x})$ and $\{\nabla \psi_n(\bar{x})\}_{n\in I}$ are linearly independent {\bf(d)}.
Then, there exist $\lambda\in\mathbb{R}$ and $\{\alpha_n\}_{n\in I}$ with $\alpha_n\ge0$ for any $n\in I$ such that
\begin{equation}\label{eq:KKT}
\nabla f(\bar{x}) - \lambda\,\nabla \phi(\bar{x}) + \sum_{n\in I}\alpha_n\,\nabla \psi_n(\bar{x})=0\;.
\end{equation}
\end{thm}

\bibliography{biblio}
\bibliographystyle{plain}
\end{document}